\documentclass[11pt,envcountsame]{article}
\usepackage{amsmath,amssymb,amsthm,thmtools}

\usepackage{microtype}

\usepackage{fullpage}

\usepackage{url}
\newtheorem{definition}{Definition}

\newtheorem{lemma}[definition]{Lemma}
\newtheorem{theorem}[definition]{Theorem}
\newtheorem{lemma1}[definition]{Lemma}

\newtheorem{corollary}[definition]{Corollary}
\newtheorem{example}[definition]{Example}
\newtheorem{conjecture}[definition]{Conjecture}

\usepackage{enumerate}
\usepackage{mathtools}
\usepackage{xspace}
\usepackage{thm-restate}
\usepackage{authblk}

\newcommand{\pcclone}[1]{\ensuremath{\langle #1 \rangle_{\not \exists}}}

\newcommand{\ar}{\ensuremath{\mathrm{ar}}}

\newcommand{\problemFont}[1]{\textsc{#1}}
\newcommand{\SAT}{\protect\ensuremath\problemFont{SAT}}

\newcommand{\sat}[1]{\ensuremath{\Gamma^{\scriptscriptstyle #1}_{\mathit{\scriptscriptstyle \mathrm{SAT}}}}}

\newcommand{\CSP}{\protect\ensuremath\problemFont{CSP}}
\newcommand{\cclone}[1]{\ensuremath{\langle #1 \rangle}}
\newcommand{\eq}{\ensuremath{{\rm Eq}}}

\newcommand{\ppol}{\ensuremath{{\rm pPol}}}
\newcommand{\pol}{\ensuremath{{\rm Pol}}}
\newcommand{\inv}{\ensuremath{{\rm Inv}}}
\newcommand{\B}{\ensuremath{\mathbb{B}}}
\newcommand{\Rdddp}{\ensuremath{R^{\scriptscriptstyle \neq \neq \neq 
    0 1}_{\scriptscriptstyle 1/3}}}
\newcommand{\Rddd}{\ensuremath{R^{\scriptscriptstyle \neq \neq \neq}_{\scriptscriptstyle 1/3}}}

\newcommand{\RD}[1]{\ensuremath{S_{#1}}}
\newcommand{\RB}{\ensuremath{S_{\mathbb{B}}}}

\newcommand{\pro}{\ensuremath{{\mathrm{Proj}}}}

\newcommand{\reduces}{\ensuremath{\leq^{\mathrm{CV}}}}
\newcommand{\setcolumns}{\ensuremath{\mathop{SetCols}}}

\newcommand{\domain}{\ensuremath{{\mathrm{dom}}}}
\newcommand{\const}[1]{\ensuremath{R^{#1}}}

\begin{document}
\title{Time Complexity of Constraint Satisfaction via Universal Algebra}

\author[1]{Peter Jonsson\thanks{peter.jonsson@liu.se}}
\author[2]{Victor
  Lagerkvist\thanks{victor.lagerqvist@tu-dresden.de}}
\author[3]{Biman Roy\thanks{biman.roy@liu.se}}
\affil[1]{\small Department of Computer and Information Science, Link\"oping
  University, Link\"oping, Sweden}
\affil[2]{\small Institut f\"ur Algebra, TU Dresden, Dresden, Germany}
\affil[3]{\small Department of Computer and Information Science, Link\"oping
  University, Link\"oping, Sweden}

\date{}
\maketitle

\abstract{The {\em
    exponential-time hypothesis} (ETH) states that $3$-SAT is 
  not solvable in subexponential time, i.e.\ not solvable in $O(c^{n})$ time for
  arbitrary $c > 1$, where $n$ denotes the number of variables.
  Problems like $k$-SAT can be viewed as special cases of the {\em
    constraint satisfaction problem} (CSP), which is the problem of determining
  whether a set of constraints is satisfiable. In this paper we study 
  the worst-case time complexity of NP-complete
  CSPs. Our main interest is in the CSP problem parameterized by a constraint
  language $\Gamma$ (CSP$(\Gamma)$), and how the choice of $\Gamma$
  affects the time complexity.
  It is believed that $\CSP(\Gamma)$ is
  either tractable or NP-complete, and the {\em algebraic CSP dichotomy
    conjecture} gives a sharp delineation of these two classes based on
  algebraic properties of constraint languages. Under this conjecture
  and the ETH, we first rule out the existence of subexponential
  algorithms for finite-domain NP-complete CSP$(\Gamma)$ problems. This
  result also extends to certain infinite-domain CSPs and structurally
  restricted $\CSP(\Gamma)$ problems. We then begin a study of the
  complexity of NP-complete CSPs where one is allowed to arbitrarily
  restrict the values of individual variables, which is a very well-studied
  subclass of CSPs. For such CSPs with finite
  domain $D$, we identify a relation $\RD{D}$ such that
  (1) CSP$(\{\RD{D}\})$ is NP-complete and (2) if $\CSP(\Gamma)$ over
  $D$ is
  NP-complete and solvable in $O(c^n)$ time, then $\CSP(\{\RD{D}\})$ is
  solvable in $O(c^n)$ time, too. Hence, the time complexity of
  $\CSP(\{\RD{D}\})$ is a lower bound for all CSPs of this particular
  kind. We also prove that the complexity of $\CSP(\{\RD{D}\})$ is
  decreasing when $|D|$ increases, unless the ETH is
  false. This implies, for instance, that for every $c>1$ there
  exists a finite-domain
  $\Gamma$  such that CSP$(\Gamma)$ is NP-complete and solvable in $O(c^{n})$ time.}

\section{Introduction}
The {\em constraint satisfaction problem} over a constraint language
$\Gamma$ ($\CSP(\Gamma)$) is the computational decision problem of
verifying whether a set of constraints over $\Gamma$ is
satisfiable or not. This problem is widely studied from both a theoretical
and a practical standpoint. From a practical point of view this
problem can be used to model many natural problems occurring in
real-world applications. From a more theoretical point of view the
$\CSP$ problem is (among several other things) 
of great interest due to its connections with {\em
universal algebra}. It is widely believed that finite-domain CSP
problems admit a dichotomy between tractable and NP-complete problems,
and the so-called {\em algebraic approach} has been used to conjecture an exact
borderline between tractable and NP-complete
problems~\cite{bulatov2005}. This conjectured borderline is sometimes
called the {\em algebraic CSP dichotomy conjecture}. 
The gist of the algebraic approach is to associate an algebra, a set
of functions satisfying a certain closure property, to each constraint
language. This associated algebra is usually referred to as the {\em
polymorphisms} of a constraint language, and is known to determine
the complexity of a $\CSP$ problem up to polynomial-time many-one
reductions~\cite{jeavons1997}.  However, the mere fact that two CSPs
are polynomial-time interreducible does not offer much insight into
their relative worst-case time complexity. For example, on the one
hand, it has been conjectured that the Boolean satisfiability problem
with unrestricted clause length, SAT, is not solvable strictly faster
than $O(2^n)$, where $n$ denotes the number of
variables~\cite{impagliazzo2001}. On the other hand, $k$-SAT is known
to be solvable strictly faster than $O(2^n)$ for every $k \geq
1$~\cite{hertli2014}, and even more efficient algorithms are known for
severely restricted satisfiability problems such as
1-in-3-SAT~\cite{wahlstrom2007}. This discrepancy in complexity stems
from the fact that a polynomial time reduction can change the
structure of an instance and e.g.\ introduce a large number of
fresh variables.
Hence, it is worthwhile to study the complexity of NP-complete CSPs
using more fine-grained notions of reductions.
To make this a bit more precise, given a constraint language
$\Gamma$ we let \[{\sf T}(\Gamma) = \inf\{c \mid \CSP(\Gamma) \textrm{ is solvable in
  time } 2^{cn}\}\]
where $n$ denotes the number of variables. 
If ${\sf T}(\Gamma) = 0$ then $\CSP(\Gamma)$ is said to be solvable in {\em
subexponential time}, and the conjecture that 3-SAT is not solvable in
subexponential time is known as the {\em exponential-time
  hypothesis} (ETH)~\cite{impagliazzo2001}. It is worth remarking that
no concrete values of ${\sf T}(\Gamma)$ are known when $\CSP(\Gamma)$
is NP-complete. Despite
this, studying properties of the function ${\sf T}$ can still be of great
interest since such properties can be used to compare and relate the worst-case
running times of NP-complete $\CSP$ problems. Moreover,
for Boolean constraint languages, several properties of the function
${\sf T}$ are known.
For example, it is
known that there exists a finite Boolean constraint language $\Gamma$
such that $\CSP(\Gamma)$ is NP-complete and ${\sf T}(\Gamma) = 0$ if
and only if ${\sf T}(\Gamma) = 0$ for {\em every} Boolean constraint
language $\Gamma$~\cite{jonsson2017}. Hence, even though the status of the ETH is unclear
at the moment, finding a subexponential time algorithm for one
NP-complete Boolean $\CSP$ problem is tantamount to being able to
solve every Boolean $\CSP$ problem in subexponential time. It is also
known that there exists a Boolean relation $R$ such that $\CSP(\{R\})$
is NP-complete but ${\sf T}(\{R\}) \leq {\sf T}(\Gamma)$ for {\em
every} Boolean constraint language $\Gamma$ such that $\CSP(\Gamma)$
is NP-complete. In Jonsson et al.~\cite{jonsson2017} this problem is
referred to as the {\em easiest NP-complete Boolean CSP problem}. The
existence of this relation e.g.\ rules out the possibility that 
for each Boolean constraint language $\Gamma$ there exists $\Delta$ such
that ${\sf T}(\Delta) < {\sf T}(\Gamma)$ --- a scenario which
otherwise would have been compatible with the ETH. These results were
obtained by considering more refined algebras than polymorphisms,
so-called {\em partial polymorphisms}. We will describe this algebraic
approach in greater detail later on, but the most important property
is that the partial polymorphisms of finite constraint languages
give rise to a partial order $\sqsubseteq$ with the property that if $\Gamma \sqsubseteq
\Delta$, then ${\sf T}(\Gamma) \leq {\sf T}(\Delta)$. We remark that
partial polymorphisms are not only useful when studying CSPs with this
very fine-grained notion of complexity, but have also
been used to study the classical complexity of many different computational problems where
polymorphisms are not applicable~\cite{herman2015,herman2016,bohler2002,bulatov2012,ham2017}.

Hence, even though no concrete values are known for ${\sf T}(\Gamma)$
when $\CSP(\Gamma)$ is NP-complete, quite a lot is known concerning
the relationship between ${\sf T}(\Gamma)$ and ${\sf T}(\Delta)$ for
Boolean $\Gamma$ and $\Delta$. In this paper we study similar
properties of the function ${\sf T}$ for constraint languages defined
over arbitrary finite domains. 
After having introduced the necessary definitions in
Section~\ref{section:preliminaries}, in Section~\ref{section:eth}
we consider the existence of subexponential time algorithms
for NP-complete CSP problems, in light of the ETH and the algebraic
CSP dichotomy conjecture. For this question we obtain a complete
understanding and prove that, assuming the algebraic CSP dichotomy
conjecture, the ETH is false if and only if (1) there exists a finite
constraint language $\Gamma$ over a finite domain such that
$\CSP(\Gamma)$ is NP-complete and ${\sf T}(\Gamma) = 0$, if and only
if (2) ${\sf T}(\Gamma) = 0$ for every finite constraint language
$\Gamma$ defined over a finite domain. In other words, finding a
subexponential time algorithm for a single NP-complete, finite-domain
$\CSP$ problem is tantamount to being able to solve all $\CSP$
problems in subexponential time. We also study structurally restricted
$\CSP$s where the maximum number of constraints a variable may
appear in is bounded by a constant $B$ ($\CSP(\Gamma)$-$B$). For
problems of this form our results are not as sharp, but we prove that,
again assuming the algebraic CSP dichotomy conjecture, that if
$\CSP(\Gamma)$ is NP-complete and $\Gamma$ satisfies an additional
algebraic condition, then there exists a constant $B$ such that
$\CSP(\Gamma)$-$B$ is not solvable in subexponential time (unless the
ETH is false). We also
remark that our proof extends to certain constraint languages defined
over infinite domain, and give several examples of infinite-domain
NP-complete CSP problems that are not solvable in subexponential time,
unless the ETH is false. These results may be interesting to compare
to those of De Haan et al.~\cite{szeider2015}, who study 
subexponential algorithms for structurally
restricted CSPs. One crucial difference to our results is that De Haan
et al.\ do not consider constraint language restrictions. For example,
it is proven that $\CSP(\Delta)$-$2$, where $\Delta$ is the set of all
finitary relations of finite cardinality, is not solvable in subexponential
time unless the ETH is false. However, a result of this form tells us
very little about the complexity of $\CSP(\Gamma)$-2 for specific
constraint languages, since it does not imply that
$\CSP(\Gamma)$-2 is not solvable in subexponential time for every
$\Gamma$ such that $\CSP(\Gamma)$-2 is NP-complete.

We have thus established that ${\sf T}(\Gamma) > 0$ for every
NP-complete, finite-domain $\CSP(\Gamma)$, assuming the ETH and the
algebraic CSP dichotomy conjecture. This immediately raises the
question of which further insights can be gained concerning the
behaviour of the function ${\sf T}$.  For example, for a fixed finite
domain, is it possible to construct an infinite chain of NP-complete
CSPs with strictly decreasing complexity such that ${\sf T}$ tends to
0?  We study such questions in Section~\ref{section:easiest_csp} for
CSPs where one in an instance is allowed to restrict the values of
individual variables arbitrarily. This restricted CSP problem is
particularly well-studied, and it is used as {\em the} definition of
CSPs in many cases: see, for instance, the textbook by Russell and
Norvig~\cite[Section 3.7]{DBLP:books/daglib/0023820} and the handbook
by Rossi et al.~\cite[Section 2]{DBLP:reference/fai/2}.  This may be
viewed as restricting oneself to constraint languages that contain all
unary relations.  A closely related restriction (that is typically
used when studying CSPs from the algebraic viewpoint) is that every
unary relation is primitively positively definable in $\Gamma$ (see
Section~\ref{section:preliminaries}).  Such constraint languages are
known as {\em conservative}.  These two restrictions are
computationally equivalent up to polynomial-time many-one reductions
but it is not known whether they are equivalent under reductions that
preserve time complexity.  Thus, we need to separate them, so we say
that a constraint language that contains all unary relations is {\em
ultraconservative}.  We note that the algebraic CSP dichotomy
conjecture has been verified to hold for the conservative
CSPs~\cite{bulatov2011} so it holds for ultraconservative CSPs, too.
We show that for every finite domain $D$ there exists a relation
$\RD{D}$ such that $\CSP(\{\RD{D}\})$ is NP-complete and ${\sf
T}(\{\RD{D}\}) = {\sf T}(\{\RD{D}\} \cup 2^D) \leq {\sf T}(\Gamma)$
for every ultraconservative and NP-complete $\CSP(\Gamma)$ over $D$. This
relation will be formally defined in Section~\ref{sec:extensions}, but
is worth pointing out that $\RD{D}$ contains only three tuples and
that $\CSP(\{\RD{D}\})$ can be viewed as a higher-domain variant of the
monotone 1-in-3-SAT problem.  We refer to $\CSP(\{\RD{D}\} \cup 2^D)$
as the {\em easiest NP-complete ultraconservative CSP problem over
$D$}\footnote{Note that $2^D$ is the set of all unary relations over $D$.}. Note that the properties of the relation $\RD{D}$ rule out the
possibility of an infinite sequence of ultraconservative languages
$\Gamma_1, \Gamma_2, \ldots$ such that each $\CSP(\Gamma_i)$ is
NP-complete and ${\sf T}(\Gamma_i)$ tends to 0, but also have stronger
implications, since the value ${\sf T}(\{\RD{D}\})$ is a conditional
lower bound for the complexity of {\em all} NP-complete,
ultraconservative CSPs over $D$.

To prove these results we have to overcome several major obstacles.
Similar to Jonsson et al.~\cite{jonsson2017}) we use partial
polymorphisms instead of total polymorphisms in order to achieve more
fine-grained notions of reductions. However, the proof strategy used
in Jonsson et al.~\cite{jonsson2017} does not work for arbitrary
finite domains since it requires a comprehensive understanding of the
polymorphisms of constraint languages resulting in NP-complete CSPs,
which is only known for the Boolean domain~\cite{pos41}.  Our first
observation to tackle this difficulty is that the reformulation of
conservative CSP dichotomy theorem making use of {\em primitive
positive interpretations} (pp-interpretations) is useful in our
context. At the moment, we may think of a pp-interpretation as a tool
which allows us to compare the expressitivity of constraint languages
defined over diferent domains, modulo logical formulas consisting of
existential quantification, conjunction, and equality constraints. It
is well-known that pp-interpretations can be used to obtain
polynomial-time reductions between CSPs, and that a conservative
$\CSP(\Gamma)$ problem is NP-complete if and only if $\Gamma$
pp-interprets 3-SAT~\cite{barto2014,bulatov2011}.  However, as already
pointed out, such reductions are not useful when studying CSPs with
respect to the function ${\sf T}$, and it is a priori not evident how
the assumption that $\Gamma$ can pp-interpret 3-SAT can be used to
show that ${\sf T}(\{\RD{D}\}) \leq {\sf T}(\Gamma)$. Using properties of conservative
constraint languages and quantifier-elimination techniques we in
Section~\ref{sec:extensions} first show
that this assumption can be used to prove there exists a
relation $R$ over $D$ of cardinality 3 such that (1) $\CSP(\{R\})$ is
NP-complete and (2) ${\sf T}(\{R\}) \leq {\sf T}(\Gamma)$. However, this is
not enough in order to isolate a unique easiest problem, since there
for every finite domain exists a large number of such relations. In
Section~\ref{section:extension_complexity}, using a combination of
partial clone theory and size-preserving reductions, we show that ${\sf
T}(\{\RD{D}\}) \leq {\sf T}(\{R\})$ for every such relation $R$ of
cardinality 3. We then analyse the time complexity of the
problem $\CSP(\{\RD{D}\})$ and prove that ${\sf T}(\{\RD{D}\})$ tends
to 0 for increasing values of $|D|$. This also shows, despite the fact
that no finite-domain NP-complete $\CSP(\Gamma)$ is solvable in
subexponential time (if the algebraic CSP dichotomy conjecture and
the ETH are true), that one for every $c > 0$ can find $\Gamma$ over a
finite domain such that $\CSP(\Gamma)$ is NP-complete and solvable in
$O(2^{cn})$ time.
When all of these results are adjoined, they demonstrate that the function
${\sf T}$ can indeed be analysed without an extensive knowledge of the
polymorphisms related to a constraint language.

\section{Preliminaries}
\label{section:preliminaries}

\noindent
\underline{{\bf Relations and constraint languages.}}
A $k$-ary {\em relation} $R$ over a set $D$ is a subset of $D^k$, and
we write $\ar(R) = k$ to denote its arity. A finite set of relations $\Gamma$ over a set $D$ is called a {\em constraint
language}. 
Given two tuples $s$ and $t$ we let $s^{\frown} t$ denote the
concatenation of $s$ and $t$, i.e., if $s = (s_1, \ldots, s_{k_1})$
and $t = (t_1, \ldots, t_{k_2})$ then $s^{\frown}t = (s_1, \ldots,
s_{k_1}, t_1, \ldots, t_{k_2})$. 
If $t$ is an $n$-ary tuple we let $t[i]$ denote its $i$th element and $\pro_{i_1, \ldots, i_{n'}}(t) = (t[i_1], \ldots, t[i_{n'}])$, $n' \leq n$,
denote the {\em projection} of $t$ on the coordinates $i_1, \ldots,
i_{n'} \in \{1, \ldots, n\}$. Similarly, if $R$ is an $n$-ary relation
we let $\pro_{i_1, \ldots, i_{n'}}(R) = \{\pro_{i_1, \ldots,
  i_{n'}}(t) \mid t \in R\}$. 
We write $\eq_D$ for the equality relation $\{(x,x) \mid x \in D\}$. If there is no risk for confusion we omit the
subscript and simply write $\eq$. For each $d \in D$ we write
$\const{d}$ for the unary, constant relation $\{(d)\}$. We will occasionally represent
relations by first-order formulas, and if $\varphi(x_1, \ldots, x_k)$
is a first-order formula with free variables $x_1, \ldots, x_k$ then
we write $R(x_1, \ldots, x_k) \equiv \varphi(x_1, \ldots, x_k)$ to
define the relation $R = \{(f(x_1), \ldots, f(x_k)) \mid f$ is a model
of $\varphi(x_1,\ldots,x_k)\}$. As a graphical representation, we
will sometimes view a $k$-ary relation $R = \{t_1, \ldots,
t_m\}$ as an $m \times k$ matrix where the columns of the matrix
enumerate the arguments of the relation (in some fixed ordering). For
example, 
$
\bigl(
\begin{smallmatrix}
0 & 0 & 1 & 1\\
0 & 1 & 0 & 1 
\end{smallmatrix} \bigr)
$
represents the relation $\{(0,0,1,1),(0,1,0,1)\}$.

\medskip

\noindent
\underline{{\bf The constraint satisfaction problem.}}
The {\em constraint satisfaction problem} over a
constraint language $\Gamma$ over $D$
($\CSP(\Gamma)$) is the computational decision problem defined as
follows.
\smallskip

\noindent
{\sc Instance:} A set $V$ of variables and a set $C$ of constraint
applications $R(x_1,\ldots,x_{k})$ where $R \in \Gamma$, $\ar(R) = k$,
and $x_1,\ldots, x_{k} \in V$.

\noindent
{\sc Question:} Does there exist $f : V \rightarrow D$ such
that $(f(x_1),\ldots,f(x_{k})) \in R$ for each 
$R(x_1,\ldots,x_{k})$~in~$C$?

\smallskip
If $\Gamma = \{R\}$ is singleton then we write $\CSP(R)$ instead of
$\CSP(\{R\})$, and if $\Gamma$ is Boolean we
typically write $\SAT(\Gamma)$ instead of $\CSP(\Gamma)$. We
let $\B = \{0,1\}$. For example, let $\Rdddp =
\{(0,0,1,1,1,0,0,1), (0,1,0,1,0,1,0,1), (1,0,0,0,1,1,0,1)\}$. The
$\SAT$ problem over $\Rdddp$ can be seen as a variant of 1-in-3-SAT
where each variable in each constraint has a complementary
variable. We will return to this $\SAT$ problem several times in the
sequel. For each $k \geq 3$ let $\sat{k}$ be the constraint language
which for every $t \in \B^{k}$ contains the relation $\B^{k}
\setminus \{t\}$. Hence, $\SAT(\sat{k})$ can be viewed as an
alternative formulation of $k$-SAT.

\medskip

\noindent
\underline{{\bf Primitive positive definitions and interpretations.}}
Let $\Gamma$ be a constraint language. A $k$-ary relation $R$ is said
to have a {\em primitive positive definition} (pp-definition) over
$\Gamma$ if $R(x_1, \ldots, x_{k}) \equiv \exists y_1, \ldots,
y_{k'}\, . \, R_1(\mathbf{x_1}) \wedge \ldots \wedge
R_m({\mathbf{x_m}}),$ where each $R_i \in \Gamma \cup \{\eq\}$ and
each $\mathbf{x_i}$ is an $\ar(R_i)$-ary tuple of variables over
$x_1,\ldots, x_{k}$, $y_1, \ldots, y_{k'}$. In addition, if the
primitive positive formula does not contain any existentially
quantified variables, we say that it is a {\em quantifier-free
primitive positive formula} (qfpp), and if it does not contain any
equality constraints we say that it is a {\em equality-free primitive
positive formula} (efpp). For example, the reader can verify that the
textbook reduction from $k$-SAT to $(k-1)$-SAT, where a clause of
length $k$ is replaced by clauses of length $k-1$ making use of one
fresh variable, can be formulated as a pp-definition but not as a qfpp-definition.
We write $\cclone{\Gamma}$ (respectively
$\pcclone{\Gamma}$) to denote the smallest set of relations containing
$\Gamma$ and which is closed under pp-definitions (respectively
qfpp-definitions). If $\Gamma = \{R\}$ is singleton then we instead
write $\cclone{R}$ and $\pcclone{R}$. Note that $\cclone{\Gamma}$ is
closed under projections, in the sense that if $R \in \cclone{\Gamma}$
then $\pro_{i_1, \ldots, i_n}(R) \in \cclone{\Gamma}$ for all $i_1, \ldots, i_n \in \{1,
\ldots, \ar(R)\}$, but that this does not necessarily hold for
$\pcclone{\Gamma}$.  
Jeavons~\cite{Jeavons1998} proved the following important result.

\begin{theorem} \label{ppdefcomplexity}
If $\Gamma$ is a constraint language and $\Delta$ is a finite
subset of $\cclone{\Gamma}$, then CSP$(\Delta)$ is polynomial-time
reducible to CSP$(\Gamma)$.
\end{theorem}

Theorem~\ref{ppdefcomplexity} naturally holds also for relations
defined by qfpp- or efpp-formulas.  However, there are additional
advantages of these more restricted ways of defining relations and we
will return to them later on.  We are now ready to
define the concept of primitive positive interpretations.

\begin{definition} \label{def:ppi}
  Let $D$ and $E$ be two domains and let $\Gamma$ and $\Delta$
  be two constraint languages over $D$ and $E$, respectively. A {\em
    primitive positive interpretation} (pp-interpretation) of $\Delta$
  over $\Gamma$ consists of a $d$-ary relation $F \subseteq D^{d}$ and
  a surjective function $f : F \rightarrow E$ such that
    $F, f^{-1}(\eq_E) \in \cclone{\Gamma}$ and
    $f^{-1}(R) \in \cclone{\Gamma}$ for every $R \in \Delta$, 
  where $f^{-1}(R)$, $\ar(R) = k$, denotes the $(k \cdot d)$-ary relation
  \[
\{(x_{1,1}, \ldots, x_{1,d}, \ldots, x_{k, 1},
  \ldots, x_{k, d})  \in D^{k \cdot d} \mid (f(x_{1,1},
  \ldots, x_{1,d}), \ldots, f(x_{k, 1}, \ldots, 
  x_{k, d})) \in R\}.
\]
\end{definition}

 The main purpose of pp-interpretations is to relate constraint
 languages which might be incomparable
 with respect to pp-definitions. For an example, let us consider the
 relation $R_{\neq} = \{(x,y) \in \{0,1,2\}^2 \mid x \neq
 y\}$, and observe that $\CSP(\{R_{\neq}\})$ corresponds to the
 3-coloring problem. We invite the reader to verify that
 the standard reduction from 3-coloring to 3-SAT can be phrased as a
 pp-interpretation of $R_{\neq}$ over $\sat{3}$, but that this
 reduction cannot be expressed via pp-definitions due to the different domains.
  Hence, pp-interpretations are generalizations of pp-definitions, and can be used to obtain polynomial-time reductions
  between $\CSP$s.

\begin{theorem}[cf. Theorem 5.5.6 in Bodirsky~\protect\cite{Bodirsky:habil}]
If $\Gamma,\Delta$ are constraint languages and there is a pp-interpretation
of $\Delta$ over $\Gamma$, then CSP$(\Delta)$ is polynomial-time
reducible to CSP$(\Gamma)$.
\end{theorem}

\medskip

\noindent
\underline{{\bf Polymorphisms and partial polymorphisms.}}
Let $f$ be a $k$-ary function over a finite domain $D$. We say that
$f$ is a {\em polymorphism} of an $n$-ary relation $R$ over $D$ if $f(t_1,
\ldots, t_k) \in R$ for each $k$-ary sequence of tuples $t_1, \ldots,
t_k \in R$. Here, and in the sequel, we use $f(t_1, \ldots, t_k)$ to
denote the componentwise application of the function $f$ to the tuples
$t_1, \ldots, t_k$, i.e., $f(t_1, \ldots, t_k)$ is a shorthand for
the $n$-ary tuple $(f(t_1[1], \ldots, t_k[1]), \ldots, f(t_1[n],
\ldots, t_k[n]))$. Similarly, if $f$ is a partial function over $D$,
we say that $f$ is a {\em partial polymorphism} of an $n$-ary relation
$R$ over $D$ if $f(t_1, \ldots, t_k) \in R$ for every sequence $t_1,
\ldots, t_k$ such that $f(t_1, \ldots, t_k)$ is defined for each
componentwise application. If $f$ is a polymorphism or a partial
polymorphism of a relation $R$ then we occasionally also say that $R$
is {\em invariant} under $f$. We let $\pol(R)$ and $\ppol(R)$ denote the
set of all polymorphisms, respectively partial polymorphisms, of the
relation $R$. Similarly, for a constraint language $\Gamma$, we write
$\pol(\Gamma)$ for the set $\bigcap_{R \in \Gamma} \pol(R)$, and
$\ppol(\Gamma)$ for the set $\bigcap_{R \in \Gamma} \ppol(R)$.
We write $\inv(F)$ to denote the set of all relations invariant
under the set of total or partial functions $F$. It is known that $\inv(\pol(\Gamma)) = \cclone{\Gamma}$
and that $\inv(\ppol(\Gamma)) = \pcclone{\Gamma}$, giving rise to the
following {\em Galois connections}.

\begin{theorem}[\cite{BKKR69i,BKKR69ii,Gei68,romov1981}]
  \label{theorem:galois}
  Let $\Gamma$ and $\Gamma'$ be two constraint languages. Then 
   $\Gamma \subseteq \cclone{\Gamma'}$ if and only if  $\pol(\Gamma') \subseteq
    \pol(\Gamma)$ and $\Gamma \subseteq \pcclone{\Gamma'}$ if and only if $\ppol(\Gamma') \subseteq
    \ppol(\Gamma)$.
\end{theorem}

\medskip

\noindent
\underline{{\bf Time complexity and size-preserving reductions.}}
Given a constraint language $\Gamma$ we let
${\sf T}(\Gamma) = \inf\{c \mid \CSP(\Gamma) \textrm{ is solvable in
  time } 2^{cn}\}$ where $n$ denotes the number of variables in a
given instance.
If ${\sf T}(\Gamma) = 0$ then $\CSP(\Gamma)$ is said to be solvable in
{\em subexponential time}. The conjecture that $\SAT(\sat{3}) > 0$ is
known as the {\em exponential-time hypothesis}
(ETH)~\cite{impagliazzo98}. We now introduce a type of
reduction useful for studying the complexity of CSPs 
with respect to the function ${\sf T}$.

\begin{definition}
  Let $\Gamma$ and $\Delta$ be two constraint languages. 
 The function $f$  from the instances of $\CSP(\Gamma)$ to the
 instances of $\CSP(\Delta)$ is a {\em many-one
   linear variable reduction} (LV-reduction) with parameter $d \geq 0$
 if (1) $f$ is a polynomial-time many-one reduction from
 $\CSP(\Gamma)$ to $\CSP(\Delta)$ and (2) $|V'| = d \cdot |V| + O(1)$ where $V$, $V'$ are the set of
  variables in $I$ and $f(I)$, respectively.
\end{definition}

The term CV-reduction, short for {\em constant variable reduction}, is
used to denote LV-reductions with parameter 1, and we write
$\CSP(\Gamma) \reduces \CSP(\Delta)$ when
$\CSP(\Gamma)$ has a CV-reduction to $\CSP(\Delta)$. It follows that if $\CSP(\Gamma) \reduces \CSP(\Delta)$
then ${\sf T}(\Gamma) \leq {\sf T}(\Delta)$, and if $\CSP(\Gamma)$
LV-reduces to $\CSP(\Delta)$ then ${\sf T}(\Gamma) = 0$ 
if ${\sf T}(\Delta) = 0$. We have the following theorem from Jonsson et
al.~\cite{jonsson2017}, relating the partial polymorphisms of constraint languages
with the existence of CV-reductions.

\begin{theorem}[\cite{jonsson2017}] \label{theorem:cvred}
  Let $D$ be a finite domain and let $\Gamma$ and $\Delta$ be two
  constraint languages over $D$. If $\ppol(\Delta) \subseteq
  \ppol(\Gamma)$ then $\CSP(\Gamma) \reduces \CSP(\Delta)$.
\end{theorem}
We remark that the original proof
only concerned Boolean constraint languages but that the same proof
also works for arbitrary finite domains.
Using Theorem~\ref{theorem:cvred} and algebraic techniques from
Schnoor and Schnoor~\cite{schnoor2008a}, Jonsson et
al.~\cite{jonsson2017} proved that ${\sf T}(\{\Rdddp\}) \leq {\sf
  T}(\Gamma)$ for any finite $\Gamma$ such that $\SAT(\Gamma)$ is
NP-complete. This problem was referred to as the {\em easiest
  NP-complete $\SAT$ problem}. We will not go into the details but remark
that the proof idea does not work for arbitrary finite domains since it
requires a characterisation of every $\pol(\Gamma)$ such that
$\CSP(\Gamma)$ is NP-complete. Such a list is known for the Boolean
domain due to Post~\cite{pos41} and Schaefer~\cite{sch78}, but not
for larger domains.

\label{sec:algcspconj}

\medskip

\noindent
\underline{{\bf Complexity of CSP.}}
Let $\Gamma$ be a constraint language over a finite domain $D$. We say
that $\Gamma$ is {\em idempotent} if $\const{d} \in \cclone{\Gamma}$
for every $d \in D$, {\em conservative} if $2^{D} \subseteq
\cclone{\Gamma}$, and {\em ultraconservative} if $2^{D} \subseteq \Gamma$.
A unary function $f \in \pol(\Gamma)$ is said to be an {\em
endomorphism}, and if $f$ in addition is bijective it is said to be an
{\em automorphism}. A constraint language $\Gamma$ is a {\em core} if
every endomorphism is an automorphism. The following theorem is
well-known, see e.g.\ Barto~\cite{barto2014}, but is usually expressed
in term of polynomial-time many-one reductions instead of CV-reductions.

\begin{theorem}
  \label{thm:core}
  Let $\Gamma$ be a core constraint language over the domain $\{d_0, \ldots,
  d_{k-1}\}$. Then $\CSP(\Gamma \cup \{\const{d_0}, \ldots,
  \const{d_{k-1}}\}) \reduces \CSP(\Gamma)$.
\end{theorem}

If $\Gamma$ is a constraint language over $D=\{d_0,\dots,d_{k-1}\}$, then
$\Gamma \cup \{\const{d_0}, \ldots,\const{d_{k-1}}\}$ is both idempotent
and a core since its only endomorphism is the identity function on $D$.
The {\em CSP dichotomy conjecture} states that for any
$\Gamma$ over a finite domain, $\CSP(\Gamma)$ is either
tractable or
NP-complete~\cite{FV98}. This conjecture was later refined by Bulatov
et al.~\cite{bulatov2005} to also
induce a sharp characterization of the tractable and intractable
cases, expressed in terms of algebraic properties of the constraint
language, and is usually called the {\em algebraic CSP dichotomy
  conjecture}. We will use the following variant of the conjecture
which is expressed in terms of pp-interpretations.

\begin{conjecture} \cite{barto2014,bulatov2005}
  Let $\Gamma$ be an idempotent constraint language over a finite
  domain. Then $\CSP(\Gamma)$ is NP-complete if $\Gamma$ pp-interprets
  $\sat{3}$ and tractable otherwise.
\end{conjecture}

It is worth remarking that if $\Gamma$ pp-interprets $\sat{3}$ then
$\Gamma$ can pp-interpret every finite-domain
relation~\cite[Theorem 5.5.17]{Bodirsky:habil}.

\section{Subexponential Time Complexity}
\label{section:eth}
For Boolean constraint languages it has been proven
that $\SAT(\sat{3})$ is solvable in subexponential time if and only if
there exists a finite Boolean constraint language $\Gamma$ such that
$\SAT(\Gamma)$ is NP-complete and solvable in subexponential time~\cite{jonsson2017}. We
will strengthen this result to arbitrary domains and prove that
$\CSP(\Gamma)$ is never solvable in subexponential time if $\Gamma$ can
pp-interpret $\sat{3}$, unless the ETH is false. The result can also be extended to
certain structurally restricted $\CSP$s.
The {\em degree} of a variable $x \in V$ of an instance $(V,C)$ of
$\CSP(\Gamma)$ is the number of constraints in $C$ containing $x$. We
let $\CSP(\Gamma)$-$B$, $B \geq 1$, denote the  restricted
$\CSP(\Gamma)$ problem where each variable occurring in an instance has degree at most
$B$. We then obtain the following theorem, whose proof can be found in Appendix~\ref{a0}.

\begin{restatable}{theorem}{thmsubexp} \label{thm:subexp}
Assume that the ETH is true and let $\Gamma$ be a  
finite constraint language over a domain $D$ such that  $\Gamma$
pp-interprets $\sat{3}$. Then $\CSP(\Gamma)$ is not solvable in
subexponential time, and if $\Gamma$ efpp-defines $\eq_D$
then there exists a constant $B$, depending only on $\Gamma$, such that
$\CSP(\Gamma)$-$B$ is not solvable in subexponential time.
\end{restatable}

We have now obtained a complete understanding of subexponential solvability
of finite-domain CSPs modulo the ETH. 

\begin{corollary}
  Assume that the algebraic CSP dichotomy conjecture is true. Then the
  following statements are equivalent. 
  \begin{enumerate}
  \item
    The ETH is false.
  \item
    $\CSP(\Gamma)$ is solvable in subexponential time for every
    finite $\Gamma$ over a finite domain.
  \item
    There exists a finite constraint language $\Gamma$ over a
    finite domain $D$ such that $\CSP(\Gamma)$ is NP-complete and subexponential.
  \end{enumerate}
\end{corollary}

\begin{proof}
  The implication from (1) to (2) follows from Impagliazzo et
  al.~\cite[Theorem 3]{impagliazzo98}. The implication from (2) to (3) is
  trivial. For the implication from (3) to (1), we first note that
  $\CSP(\Gamma^c) \reduces \CSP(\Gamma)$, where $\Gamma^c$ is the core
  of $\Gamma$~\cite[Theorem 3.5]{barto2014}. If $\Gamma^c$ is expanded
  with all constants, then Theorem~\ref{thm:core} shows that the
  complexity does not change, and, last, this language can pp-interpret
  $\sat{3}$, due to the assumption that the algebraic
  CSP dichotomy conjecture is true, which via 
  Theorem~\ref{thm:subexp} implies that 3-SAT is solvable in
  subexponential time, and thus that the ETH is false.
\end{proof}

For $\CSP(\Gamma)$-$B$ our results are not as precise since we need
the additional assumption that the equality relation is
efpp-definable. This is not surprising since the most powerful
dichotomy results for $\CSP$s are usually concerned with either
constraint language restrictions~\cite{bulatov2011,bulatov2005},
structural restrictions~\cite{szeider2015,grohe2006}, but rarely both
simultaneously. However, in the Boolean domain there are plenty of
examples which illustrates how the equality relation may be
efpp-defined~\cite{creignou2014,jonsson2017}, suggesting that similar
techniques may also exist for larger domains.

Theorem~\ref{thm:subexp} also applies to many interesting classes of
infinite-domain CSPs. For example, if we consider $\Gamma$ such that
each $R \in \Gamma$ has a first-order definition over the structure
$(\mathbb{Q}; <)$, it is known that $\CSP(\Gamma)$ is NP-complete if
and only if $\Gamma$ can pp-interpret
$\sat{3}$~\cite{Bodirsky:habil,bodirskykara2010}. Hence,
Theorem~\ref{thm:subexp} is applicable, implying that if
$\CSP(\Gamma)$ is not solvable in subexponential time if it is
NP-complete, unless the ETH fails. More examples of infinite-domain
CSPs where Theorem~\ref{thm:subexp} is applicable includes graph
satisfiability problems~\cite{bodirsky2015} and phylogeny
constraints~\cite{jonsson2016b}. Note that all of these results hold
independently of whether the algebraic CSP dichotomy is true or
not. We also remark that the intractable cases of the CSP dichotomy
conjecture for certain infinite-domain CSPs are all based on
pp-interpretability of $\sat{3}$~\cite{barto2016}. If this conjecture
is correct, Theorem~\ref{thm:subexp} and the ETH implies that none of
these problems are solvable in subexponential time.

\section{The Easiest NP-Complete Ultraconservative CSP Problem}
\label{section:easiest_csp}
The results from
Section~\ref{section:eth}, assuming the algebraic CSP dichotomy
conjecture and the ETH, implies that ${\sf T}(\Gamma) > 0$ for any
finite-domain and NP-complete $\CSP(\Gamma)$. However, it is safe
to say that very little is known about the behaviour of the
function ${\sf T}$ in more general terms. For example, is there for 
an arbitrary NP-complete
$\CSP(\Gamma)$ possible to find an NP-complete $\CSP(\Delta)$ such that ${\sf
  T}(\Delta) < {\sf T}(\Gamma)$? Such a scenario would be compatible
with the consequences of Theorem~\ref{thm:subexp}.
We will show that this is unlikely, and prove that there for every
finite domain $D$ exists a relation $\RD{D}$ such that $\CSP(\RD{D})$ is
NP-complete but ${\sf T}(\{\RD{D}\}) \leq
{\sf T}(\Gamma)$ for any ultraconservative $\Gamma$ over $D$ such that
$\CSP(\Gamma)$ is NP-complete. 
To prove this we
have divided this section into two parts. In
Section~\ref{sec:extensions} we show that if $\Gamma$ is
ultraconservative and $\CSP(\Gamma)$ is NP-complete, then there exists
a relation $R \in \pcclone{\Gamma}$ which shares certain properties
with the relation $\Rdddp$. In
Section~\ref{section:extension_complexity} we use properties of these relations in
order to prove that there for every finite domain $D$ is possible to find
a relation $\RD{D}$ such that $\CSP(\RD{D})$ is CV-reducible to any other
NP-complete and ultraconservative $\CSP(\Gamma)$ problem.

\subsection{$\RB$-Extensions}
\label{sec:extensions}
The columns of the matrix representation of the relation $\Rdddp$ from
Jonsson et al.~\cite{jonsson2017} (resulting in the easiest
NP-complete $\SAT$ problem) enumerates all Boolean ternary tuples.  We
generalize this relation to arbitrary finite domains as follows.

\begin{definition}
\label{definition:easiest}
  For each finite $D$ let $\RD{D} = \{t_1, t_2, t_3\}$ denote the $|D|^{3}$-ary
relation such that there for every
$(d_1, d_2, d_3) \in D^3$ exists $1 \leq i \leq |D|^3$ such that
$(t_1[i], t_2[i], t_3[i]) = (d_1, d_2, d_3)$. 
\end{definition}

Hence, similar to $\Rdddp$, the columns of the matrix representation of $\RD{D}$ enumerates
all ternary tuples over $D$. For each $D$ the relation $\RD{D}$ is unique up to permutation
of arguments, and although we will usually not be concerned with the exact
ordering, we sometimes assume that $\RB = \Rdddp$ and
that $\pro_{1, \ldots, 8}(\RD{D}) = \RB$.
The notation $\RD{D}$ is a mnemonic for {\em saturated},
and the reason behind this will become evident in
Section~\ref{section:saturation}. For example, for $\{0,1,2\}$ we
obtain a relation $\{t_1, t_2, t_3\}$ with $27$ distinct arguments
such that $(t_1[i], t_2[i], t_3[i]) \in \{0,1,2\}^3$ for each $1 \leq
i \leq 27$.
Jonsson et al.~\cite{jonsson2017} proved that
$\RB \in \pcclone{\Gamma}$ for every Boolean and idempotent
constraint language $\Gamma$ such that $\SAT(\Gamma)$ is NP-complete. This is not true for arbitrary finite
domains, and in order to prove an analogous result we will need the following definition.

\begin{definition}
  Let $R$ be an $n$-ary relation of cardinality 3 over a domain $D$,
  $|D| \geq 2$. Let $a,b \in D$ be two distinct
  values. If 
  there exists $i_1, \ldots, i_8 \in \{1, \ldots, n\}$ such that
  \[\pro_{i_1, \ldots, i_8}(R) = \{(a, a, b, b, b, a, a, b), (a, b, a, b,
  a, b, a, b), (b, a, a, a, b, b, a, b)\},\]
then we say that $R$ is an {\em $\RB$-extension}.
\end{definition}

For example, $\RD{D}$ is an $\RB$-extension for every domain $D$.
Note that $\CSP(R)$ is always NP-complete when $R$ is an
$\RB$-extension. We will now prove that if $\CSP(\Gamma)$ is
NP-complete and $\Gamma$ is ultraconservative, then $\Gamma$ can
pp-define an $\RB$-extension.

\begin{lemma1} \label{lem:ppi}
  Let $\Gamma$ be an ultraconservative constraint language over a
  finite domain $D$ such 
  that $\CSP(\Gamma)$ is
  NP-complete. Then there
  exists a relation $R \in \cclone{\Gamma}$ which is an $\RB$-extension.
\end{lemma1}

\begin{proof}
  Since $\CSP(\Gamma)$ is NP-complete and $\Gamma$ is ultraconservative, $\Gamma$ can pp-interpret every
  Boolean relation. Therefore let $f : F \rightarrow
  \B$, $F \subseteq D^d$ denote the parameters in the pp-interpretation
  of $\RB$, and note that $f^{-1}(\RB) \in \cclone{\Gamma}$, but
  that $f^{-1}(\RB)$ is not necessarily an $\RB$-extension since
  it could be the case that $|f^{-1}(\RB)| > 3$. 
  Pick two tuples $s$ and $t$ in $F$ such that $f(s) = 0$ and $f(t) = 1$. Such
  tuples must exist since $f$ is surjective. Now consider the relation $F_1(x_1, \ldots, x_d) \equiv
  F(x_1, \ldots, x_d) \land \{(s[1]), (t[1])\}(x_1) \land \ldots \land
  \{(s[d], t[d])\}(x_d)$. This relation is pp-definable over $\Gamma$ since
  $\Gamma$ is ultraconservative and since $F \in \cclone{\Gamma}$. 
  By construction, it is clear that $s, t \in F_1$. Assume
  furthermore than $|F_1| > 2$, i.e., that there exists $u \in F_1
  \setminus \{s,t\}$. Assume without loss of generality that $f(u) =
  0$, and observe that there for each $i \in \{1,\ldots, d\}$ holds
  that $u[i] \in \{s[i], t[i]\}$. We claim that there exists some $i
  \in \{1, \ldots, d\}$ such that $u[i] = t[i] \neq s[i]$. To see
  this, observe that there must exist $i$ such that $u[i] \neq
  s[i]$, since otherwise $u = s$, and it then follows that $u[i] = t[i]$.
  Construct the relation $F_2(x_1, \ldots, x_d) \equiv F_1(x_1,
  \ldots, x_d) \land \{(u[1]), (t[1])\}(x_1) \land \ldots \land
  \{(u[d]),(t[d])\}(x_d)$, and note that $F_2 \subset F_1$ since $s \notin
  F_2$.
  By repeating this procedure we will obtain a relation $F' \subseteq F$
  such that $F' = \{s_0, s_1\}$ and such that $f(s_0) = 0$, $f(s_1) =
  1$. Using the relation $F'$ we can then pp-define the relation
  \[
  \begin{aligned}
R(x_{1,1}, \ldots, x_{1,d}, \ldots, x_{8,1}, \ldots, x_{8,d}) \equiv &
f^{-1}(\RB)(x_{1,1}, \ldots, x_{1,
    d}, \ldots, x_{8,1}, \ldots, x_{8,d}) \land \\ & F'(x_{1,1}, \ldots,
  x_{1,d}) \land \ldots \land  F'(x_{8,1}, \ldots, x_{8,d}).
\end{aligned}
\]
Clearly, if $(a_{1,1}, \ldots, a_{1,d}, \ldots, a_{8,1},
\ldots, a_{8,d}) \in R$, then $(a_{i,1}, \ldots, a_{i,d}) \in \{s_0,
s_1\}$ for each $1 \leq i \leq 8$, and $(f(a_{1,1}, \ldots, a_{1,d}), \ldots, f(a_{8,1},
\ldots, a_{8,d})) \in \RB$ if and only if $(a_{1,1}, \ldots, a_{1,d}, \ldots, a_{8,1},
\ldots, a_{8,d}) \in f^{-1}(\RB)$. Since $R \subseteq
f^{-1}(\RB)$, this implies that $(f(a_{1,1}, \ldots, a_{1,d}), \ldots, f(a_{8,1},
\ldots, a_{8,d})) \in \RB$ if and only if $(a_{1,1}, \ldots, a_{1,d}, \ldots, a_{8,1},
\ldots, a_{8,d}) \in R$ and each ($a_{i,1}, \ldots, a_{i,d}) \in
\{s_0, s_1\}$. In other words each element $f(a_{i,1}, \ldots,
a_{i,d})$ in a tuple of $\RB$ uniquely correponds to $d$ arguments $a_{i,1}, \ldots,
a_{i,d}$ in the corresponding tuple of $R$, since $(a_{i,1}, \ldots,
a_{i,d}) = s_0$ if $f(a_{i,1}, \ldots, a_{i,d}) = 0$, and $(a_{i,1}, \ldots,
a_{i,d}) = s_1$ if $f(a_{i,1}, \ldots, a_{i,d}) = 1$.
It follows that \[ R = \{s_0^{\frown}s_0^{\frown}s_1^{\frown}s_1^{\frown}s_1^{\frown}s_0^{\frown}s_0^{\frown}s_1,
   s_0^{\frown}s_1^{\frown}s_0^{\frown}s_1^{\frown}s_0^{\frown}s_1^{\frown}s_0^{\frown}s_1,
   s_1^{\frown}s_0^{\frown}s_0^{\frown}s_0^{\frown}s_1^{\frown}s_1^{\frown}s_0^{\frown}s_1\},\]
   and therefore also that $R$ is an $\RB$-extension.
\end{proof}

Observe that the existence of an $\RB$-extension $R \in
\cclone{\Gamma}$ does not imply that
$\CSP(R) \reduces \CSP(\Gamma)$. To accomplish this, we need to show
that $\Gamma$ can also qfpp-define an $\RB$-extension. 

\begin{restatable}{lemma1}{lemeasiestultra}\label{lem:easiest_ultra}
  Let $\Gamma$ be an ultraconservative constraint language over a finite domain
  $D$ such that $\CSP(\Gamma)$ is NP-complete. Then there exists a relation in $\pcclone{\Gamma}$ which is
  an $\RB$-extension. 
\end{restatable}

\begin{proof}
  We provide a short sketch of the most important ideas. For the full
  proof the reader may consult Appendix~\ref{a1}. Via
  Lemma~\ref{lem:ppi} there exists an $\RB$-extension $R \in
  \cclone{\Gamma}$. It is not necessarily the case that $R \in
  \pcclone{\Gamma}$, but it is possible to construct an
  $\RB$-extension by gradually converting the pp-definition of
  $R$ over $\Gamma$ to a qfpp-definition. To do this, let $\ar(R) = n$
  and assume e.g.\ that $R'(x_1, \ldots, x_n) \equiv \exists y
  . \varphi(x_1, \ldots, x_n, y)$, where $\exists y . \varphi(x_1, \ldots, x_n,
  y)$ is a pp-formula over $\Gamma$. Consider the relation $R'(x_1, \ldots,
  x_n, y) \equiv \varphi(x_1, \ldots, x_n, y)$. This relation is
  qfpp-definable over $\Gamma$, and if $|R'|
  > 3$ (and $R'$ is not an $\RB$-extension) one can prove that there either exists a unary
  constraint $E \in \Gamma$ such that $R''(x_1, \ldots, x_n, y) \equiv
  \varphi(x_1, \ldots, x_n, y) \land E(y)$ is an $\RB$-extension, or that there exists $i
  \in \{1, \ldots, n\}$ and a relation $F \in \pcclone{\Gamma}$ such
  that $R''(x_1, \ldots, x_i, \ldots, x_n, y, z_1, \ldots, z_{\ar(F)})
  \equiv \varphi(x_1, \ldots, x_n, y) \land F(x_i, y, z_1, \ldots,
  z_{\ar(F)})$ defines an $\RB$-extension.
\end{proof}

\subsection{Properties of and Reductions between $\RB$-Extensions}
\label{section:extension_complexity}

By Lemma~\ref{lem:easiest_ultra}, we can completely concentrate on $\RB$-extensions.
We will prove that
${\sf T}(\{\RD{D}\}) \leq {\sf T}(\Gamma)$ for every
ultraconservative $\Gamma$ over $D$ such that $\CSP(\Gamma)$ is NP-complete. To prove this, we
begin in Section~\ref{section:saturation} by investigating properties
of $\RB$-extensions, which we use to simplify the total number of
distinct cases we need to consider. With the help of these results we
in Section~\ref{section:reductions} develop techniques in order
to show that $\CSP(\RD{D}) \reduces \CSP(R)$ for every
$\RB$-extension over $D$.

\subsubsection{Saturated $\RB$-Extensions}
\label{section:saturation}
In this section we simplify the number of cases we need to consider
in Section~\ref{section:reductions}. First note that
if $R = \{t_1, t_2, t_3\}$ over $D$ is a relation with
$\ar(R) > |D|^3$ then there exists $i$ and $j$ such that $(t_1[i],
t_2[i], t_3[i]) = (t_1[j], t_2[j], t_3[j])$. We say that the $j$th
argument is {\em redundant}, and it is possible to get rid of this by
identifying the $i$th and $j$th argument with the
qfpp-definition \[R'(x_1, \ldots, x_{i}, \ldots, x_{j-1}, x_{j+1},
\ldots, x_n) \equiv R(x_1, \ldots, x_i, \ldots, x_{j-1}, x_i, x_{j+1},
\ldots, x_n).\] This procedure can be repeated until no redundant
arguments exist, and we will therefore always implicitly assume that $\ar(R) \leq
|D|^3$ and that $R$ has no redundant arguments. If $R$ is an $n$-ary
$\RB$-extension then the argument $i \in \{1, \ldots, n\}$ is said to
be {\em 1-choice}, or {\em constant}, if $|\pro_i(R)| = 1$, {\em
2-choice} if $|\pro_i(R)| = 2$, and {\em 3-choice} if $|\pro_i(R)| =
3$.

\begin{definition}
  An $n$-ary $\RB$-extension $R = \{t_1, t_2, t_3\}$ is said to be
  {\em saturated} if there for each $1 \leq i \leq n$ and every function
  $\tau : \{1, 2, 3\} \rightarrow \{1, 2, 3\}$, exists $1 \leq j \leq n$
  such that $(t_{\tau(1)}[i], t_{\tau(2)}[i], t_{\tau(3)}[i]) = (t_1[j],
  t_{2}[j], t_{3}[j])$.
\end{definition}

\begin{example} \label{ex:saturation}
The relation $\RD{D}$ is saturated for every
$D$, but if we consider the relations $R$ and $R'$ defined by the matrices
$
\bigl( \begin{smallmatrix} 
0 & 0 & 1 & 1 & 1 & 0 & 0 & 0 & 1 & 2 \\
0 & 1 & 0 & 1 & 0 & 1 & 0 & 0 & 1 & 2 \\ 
1 & 0 & 0 & 0 & 1 & 1 & 2 & 0 & 1 & 2 
\end{smallmatrix} \bigr)
\text{ and } \bigl(
\begin{smallmatrix}
0 & 0 & 1 & 1 & 1 & 0 & 0 & 0 & 1 & 2 \\
0 & 1 & 0 & 1 & 0 & 1 & 1 & 0 & 1 & 2 \\ 
1 & 0 & 0 & 0 & 1 & 1 & 2 & 0 & 1 & 2
\end{smallmatrix} \bigr)
$
then neither relation is saturated. First, $R$ is not saturated since its matrix representation, for example,
does not contain the column
$(0,2,0)$. Second, $R'$ is not saturated due to the 3-choice argument
in position 7. 
\end{example}

We now prove that we without loss of
generality may assume that an $\RB$-extension is saturated.

\begin{restatable}{lemma1}{lemtwosaturated} \label{lem:2saturated}
  Let $R$ be an $\RB$-extension. Then there exists a saturated
  $\RB$-extension $R' \in \pcclone{R}$.
\end{restatable}

\begin{proof}
We provide a short proof sketch illustrating the most important
ideas. See Appendix~\ref{a1} for a full proof. Let $n = \ar(R)$ and
define $R'$ such that $\pro_{1, \ldots, n}(R') = R$, and then add the
minimum number of arguments which makes $R'$ saturated. Via
Theorem~\ref{theorem:galois} it follows that if $R' \notin
\pcclone{R}$ then this can be witnessed by a partial function $f$
preserving $R$ but not $R'$. Therefore, there exists tuples
$t'_1, t'_2, t'_3 \in R'$ such that $f(t'_1, t'_2, t'_3) \notin R'$,
but since $\pro_{1, \ldots, n}(R') = R$ and since $R'$ is saturated,
one can prove that there must exist tuples $t_1, t_2, t_3 \in R$ such
that $f(t_1, t_2, t_3) \notin R$, contradicting the assumption that $f$
preserves $R$. Hence, $R' \in \pcclone{\Gamma}$.
\end{proof}

\begin{example}
If $R$ is the relation from Example~\ref{ex:saturation}
then the saturated relation $R'$ in $\pcclone{R}$ from
Lemma~\ref{lem:2saturated} is given by 
$R' =
\bigl(
\begin{smallmatrix}
0 & 0 & 1 & 1 & 1 & 0 & 0 & 0 & 2 & 2 & 2 & 0 & 0 & 1 & 2 \\
0 & 1 & 0 & 1 & 0 & 1 & 0 & 2 & 0 & 2 & 0 & 2 & 0 & 1 & 2 \\ 
1 & 0 & 0 & 0 & 1 & 1 & 2 & 0 & 0 & 0 & 2 & 2 & 0 & 1 & 2
\end{smallmatrix} \bigr).
$
\end{example}

\subsubsection{Reductions Between $\RB$-Extensions}
\label{section:reductions}

The main result of this section 
(Theorem~\ref{thm:mainresult} and Theorem~\ref{thm:mainresult-ext}) show that
${\sf T}(\{\RD{D}\}) = {\sf T}(\{\RD{D}\} \cup 2^D) \leq {\sf T}(\Gamma)$
whenever $\Gamma$ is an ultraconservative constraint language
over $D$ such that CSP$(\Gamma)$ is NP-complete.
The result is proven by a series of CV-reductions that we present in
Lemmas~\ref{lem:3choice}--\ref{lem:3choice3}. Due to space
constraints, we only present the proof of Lemma~\ref{lem:3choice2} which
illustrates several useful techniques, and the
remaining proofs can be found in Appendix~\ref{a1}.
Before we begin, we note
that if $R$ is an $\RB$-extension over $D$ then $\{R\}$ is not necessarily a
core. For a simple counterexample, $\{\RB\}$ is not a core 
over $\{0,1,2\}$ since the endomorphism $e(0) = 0$, $e(1) =
1$, $e(2) = 0$, is not an automorphism. However, if $R$ is an $\RB$-extension and $E =
\{d_1, \ldots, d_{m}\}$ the set $\bigcup_{1 \leq i \leq \ar(R)}
\pro_i(R)$, every endomorphism $e : E
\rightarrow E$ of $R$ must be an automorphism. Hence,
Theorem~\ref{thm:core} is applicable, and we conclude that $\CSP(\{R,
\const{d_1}, \ldots, \const{d_m}\}) \reduces \CSP(R)$. When working
with reductions between $\RB$-extensions we may therefore freely make use of constant relations.
Given an instance $(V,C)$ of $\CSP(R)$, where $R$ is an
$\RB$-extension, we say that a variable $x \in V$ occurring in a
$k$-choice position in a constraint in $C$, $1 \leq k \leq 3$, is a
{\em $k$-choice variable}.

\begin{restatable}{lemma1}{lemthreechoice} 
  \label{lem:3choice}
  Let $R$ be a saturated $\RB$-extension. Then there exists a CV-reduction $f$
  from $\CSP(R)$ to $\CSP(R)$ such that for every instance $I$ of
  $\CSP(R)$, each variable in $f(I)$ occurs as a
  3-choice variable in at most one constraint. 
\end{restatable}

\begin{lemma} \label{lem:3choice2}
  Let $R$ be a saturated $\RB$-extension and let $R'$ be
  $R$ with one or more 3-choice arguments removed, such that $R'$ is
  still saturated. Then $\CSP(R) \reduces \CSP(R')$.
\end{lemma}

\begin{proof}
  Let $R = \{t_1, t_2, t_3\}$,
  $n = \ar(R)$, $n' = \ar(R')$, and assume that 
  $\pro_{1, \ldots, n'}(R) = R'$. Let $I =
  (V,C)$ be an instance of $\CSP(R)$. First apply
  Lemma~\ref{lem:3choice} in order to obtain an instance $I_1 = (V_1,
  C_1)$ of $\CSP(R)$ such that each 3-choice variable only occurs
  in a 3-choice position in a single constraint. Assume there exists
  $x \in V_1$ and two distinct 
  constraints $c,c' \in C_1$ such that $x$ occurs in positions $i \in \{n'+1,
  \ldots, n\}$ in $c$ and in a 1- or 2-choice position $j \in \{1, \ldots, n\}$ in
  $c'$. Let $S = \pro_{i}(R) \cap \pro_{j}(R)$, and note that $|S|
  \leq 2$. Assume first that $|S| = 2$, let $S = \{d_1, d_2\}$, and
  assume without loss of generality that $t_1[i] = t_1[j] = d_1$,
  $t_2[i] = t_2[j] = d_2$, and that $t_3[i] \neq t_3[j]$ (the other cases
  can be treated similarly). Since $R$ is
  saturated there exists a 2-choice argument $i' \in \{1, \ldots, n\}$ such that $t_1[i']
  = t_1[i] = t_1[j]$, $t_2[i'] = t_2[i] = t_2[j]$, and such that
  $t_3[i'] \neq t_3[i]$. Let $y$ be the variable occurring in the
  $i'$th position of $c$. Create a fresh variable $\hat x$, replace
  $x$ in position $i$ with $\hat x$, and for each constraint where $x$
  occurs as a 1- or
  2-choice variable, replace $x$  with $y$. Repeat this
  procedure until every 3-choice variable occurring in position $n'+1, \ldots,
  n$ only occurs in a single constraint, and let $I_2 = (V_2, C_2)$ be
  the resulting instance. Assume there exists $x \in V_2$ and a
  constraint $c \in C_2$ such that $x$ occurs as a 3-choice variable
  in position $i \in \{n'+1, \ldots, n\}$ and also in a distinct
  position $j \in \{1, \ldots, n\}$ in $c$. Let $L=\{t_r \mid 1
  \leq r \leq 3, t_r[i] = t_r[j]\}$. Since $R$ does not have any
  redundant arguments it must be the case that $|L| < 3$. If $|L| = 0$
  then the instance is unsatisfiable, in which case we output an
  arbitrary unsatisfiable instance, and if $|L| = 1$ it is easy to see
  that any variable occurring in $c$ can be assigned a fixed value,
  and the constraint may be removed. Therefore, assume that $|L| = 2$, and
  e.g.\ that $L = \{t_1, t_2\}$. Since $R$ is saturated there exists a
2-choice argument $j' \in \{1, \ldots, n\}$ such
that $t_1[j'] = t_2[j'] \neq t_3[j']$. Let $y$ be the
variable occurring in position $j'$ in $c$ and add the
constraint $\const{t_1[j']}(y)$. Repeat this for every variable
occurring in position $n'+1, \ldots, n$ in a constraint in $C_2$, and then replace each constraint
$R(x_1, \ldots, x_n', \ldots, x_n)$ by $R'(x_1, \ldots, x_n)$. Note
that any variable $\hat x$ introduced in the previous step of this
reduction is removed in this transformation. Hence, the reduction is
a CV-reduction.
\end{proof}

\begin{restatable}{lemma1}{lemtwochoice} \label{lem:2choice}
  Let $R$ be an $\RB$-extension and
  let $R'$ be an $\RB$-extension obtained by adding 
  additional 2-choice arguments to $R$. Then $\CSP(R') \reduces \CSP(R)$.
\end{restatable}
\begin{restatable}{lemma1}{lemthreechoicethree} \label{lem:3choice3}
  Let $R$ be a saturated $\RB$-extension over $D$ with 3-choice arguments. Then
  $\CSP(\RD{D}) \reduces \CSP(R)$.
\end{restatable}

We have thus proved the main result of this section.

\begin{theorem} \label{thm:mainresult}
Let $D$ be a finite domain and let $\Gamma$ be a finite,
ultraconservative constraint language over $D$. If $\CSP(\Gamma)$ is
NP-complete then ${\sf T}(\{\RD{D}\}) \leq {\sf T}(\Gamma)$.
\end{theorem}

\begin{proof}
  We first observe that if $R$ is an $\RB$-extension over a finite domain $D$, 
  then
  $\CSP(\RD{D}) \reduces \CSP(R)$. By Lemma~\ref{lem:2saturated} we 
  may assume that $R$ is saturated.
  If $R$ does not contain any 3-choice
  arguments we use Lemma~\ref{lem:3choice2} together
  with Lemma~\ref{lem:2choice} and obtain a CV-reduction from
  $\CSP(\RD{D})$ to $\CSP(R)$. Hence, assume that $R$ contains one or
  more 3-choice arguments. In this case we use
  Lemma~\ref{lem:3choice3} and obtain a CV-reduction from
  $\CSP(\RD{D})$ to $\CSP(R)$. 
  By Lemma~\ref{lem:easiest_ultra} there exists an
  $\RB$-extension $R \in \pcclone{\Gamma}$, implying that
  $\CSP(R) \reduces \CSP(\Gamma)$ via Theorem~\ref{theorem:cvred}, and
  we know that $\CSP(\RD{D}) \reduces \CSP(R)$.
  We conclude that
  ${\sf T}(\{\RD{D}\}) \leq {\sf T}(\{R\}) \leq {\sf T}(\Gamma)$. 
\end{proof}

Clearly, $\{\RD{D}\}$ is not an ultraconservative constraint language
but the complexity of $\CSP(\RD{D})$ does
not change when we expand the language by adding all unary relations
over $D$ (the proof can be found in Appendix~\ref{a1}).

\begin{restatable}{theorem}{thmmainresultext} \label{thm:mainresult-ext}
  Let $D$ be a finite domain. Then ${\sf T}(\{\RD{D}\}) = {\sf
    T}(\{\RD{D}\} \cup 2^D)$.
\end{restatable}

Thus, no NP-complete CSP over
an ultraconservative constraint language over $D$ is solvable strictly
faster than $\CSP(\RD{D})$, and, in particular, ${\sf
  T}(\{\RD{D'}\}) \leq {\sf T}(\{\RD{D}\})$ whenever $D' \supseteq D$. This raises the
question of whether ${\sf T}(\RD{D}) = {\sf T}(\RD{D'})$ for all $D,D'
\supseteq \{0,1\}$, or if it is possible to find $D$ and $D'$ such
that ${\sf T}(\{\RD{D'}\}) < {\sf T}(\{\RD{D}\})$. As the following
theorem shows, this is indeed the case, unless ${\sf T}(\{\RD{D}\}) =
0$ for every finite $D$ and the ETH fails.

\begin{restatable}{theorem}{thmincreasing} \label{thm:increasing}
$\inf \{{\sf T}(\{\RD{D}\}) \mid \mbox{$D$ finite and $|D| \geq 2$}\}=0$.  
\end{restatable}
\begin{proof}
Let $D_k=\{0,\dots,k-1\}$, $k \geq 5$. 
We will analyse a simple algorithm for $\CSP(\RD{D_k})$.
Let $I=(V,C)$
be an arbitrary instance of $\CSP(\RD{D_k})$.
Extend the instance with variables $Z=\{z_0,\dots,z_{k-1}\}$ and
the constraints $\const{i}(z_i)$, $0 \leq i \leq k-1$.
Arbitrarily choose a constraint $c=\RD{D_k}(x_1,\dots,x_{k^3})$ and
let $X = \{x_1,\dots,x_{k^3}\}$.
It is straightforward to verify that if a variable $x$ appears in $k^2+1$ or
more positions, then $c$ cannot be satisfied. 
Thus, $|X| \geq k$.
If $X \cap Z = \emptyset$, then
we branch
on the three tuples in $\RD{D_k}$ and in each branch at least $k$ variables
in $V \setminus Z$ 
will be given fixed values. If a variable, say $x_i$, is given
the fixed value $d$, then we identify $x_i$ with $z_d$. Thus, at
least $k$ variables in $V \setminus Z$ are removed.
Assume to the contrary that
$X \cap Z \neq \emptyset$.
If
a variable $z \in Z$
occurs in a 3-choice position, then
the variables in $X \setminus Z$ 
can be assigned
fixed values and no branching is needed.
If no variable $z \in Z$ occurs in a 3-choice position, then
there are $k(k-1)(k-2)$
3-choice positions in $\RD{D_k}$ and they
are all covered by variables in $V \setminus Z$.
Thus, we perform three branches based on the tuples in $\RD{D_k}$.
Recall that a variable can occur in at most $k^2$ positions in
the constraint $c$ since $c$ is otherwise not satisfiable.
This implies that
at least $\lfloor \frac{k(k-1)(k-2)}{k^2} \rfloor \geq 1$
variables in $V \setminus Z$ are given fixed values (and are removed from
$V \setminus Z$) in each branch.
When there are no $\RD{D_k}$ constraints left, we check whether the
remaining set of unary constraints are satisfiable or not. It is
straightforward to perform this test in polynomial time.
A recursive equation that gives an upper bound on the time complexity of this
algorithm
is thus $T(1)=poly(||I||),T(n)=3T(n-\lfloor \frac{k(k-1)(k-2)}{k^2} \rfloor)+poly(||I||))$ (where $n$ denotes
the number of variables and $||I||$ the number of bits required to
represent $I$)
so $T(n) \in O(3^{n \cdot \frac{k^2}{k(k-1)(k-2)}
} \cdot poly(||I||))$.
The function $\frac{k^2}{k(k-1)(k-2)}$ obviously tends
to 0 with increasing $k$ so the infimum of the set $\{{\sf
  T}(\{\RD{D}\}) \mid \mbox{$D$ is finite and $|D| \geq 2$}\}$ is equal
to 0.
\end{proof}

\section{Concluding Remarks and Future Research} 
In this paper we have studied the time complexity of NP-complete
CSPs. Assuming the algebraic CSP dichotomy conjecture, we have ruled
out subexponential time algorithms for NP-complete, finite-domain
CSPs, unless the ETH is false. This proof also extends to
degree-bounded CSPs and many classes of CSPs over infinite domains. We
then proceeded to study the time complexity of CSPs over
ultraconservative constraint languages, and proved that no such
NP-complete CSP is solvable strictly faster than ${\sf
  T}(\{\RD{D}\})$. These results raise several directions for
future research.

\smallskip

\noindent
{\bf Structurally restricted CSPs and the ETH.}
Theorem~\ref{thm:subexp} shows that the algebraic approach is viable
for analysing the existence of subexponential algorithms for certain
structurally restricted $\CSP(\Gamma)$ problems. An interesting
continuation would be to try to determine which of the structurally
restricted (but not constraint language restricted) CSPs investigated
by De Haan et al.~\cite{szeider2015} could be used to prove similar
results. For example, is it the case that $\CSP(\Gamma)$ is not
solvable in subexponential time whenever $\CSP(\Gamma)$ is NP-complete
and the primal treewidth of an instance is bounded by $\Omega(n)$,
unless the ETH fails?

\smallskip

\noindent
{\bf The CSP dichotomy conjecture.} Several independent solutions to the algebraic CSP dichotomy conjecture have
recently been
announced~\cite{bulatov2017,Rafiey:etal:2017,zhuk2017}. If any of these proposed proofs is correct, it
is tempting to extend Theorem~\ref{thm:mainresult} to constraint languages that are
not necessarily ultraconservative or conservative. As a starting
point, one could try to strengthen the results in Section~\ref{sec:extensions}, in
order to prove that $\pcclone{\Gamma}$ contains an $\RB$-extension
whenever $\CSP(\Gamma)$ is NP-complete and $\Gamma$ is conservative
(but not ultraconservative).

\section*{Acknowledgements}
We thank Hannes Uppman for several helpful discussions on the topic of
this paper.
The second author has received funding from the DFG-funded
project ``Homogene Strukturen, Bedingungserf\"ullungsprobleme, und
topologische Klone'' (Project number 622397).
The third author is partially supported by the {\em National Graduate
School in Computer Science} (CUGS), Sweden.

\bibliography{references}
\bibliographystyle{plain}

\huge
\noindent
{\bf Appendix}
\normalsize

\appendix
\section{Additional Proofs for Section~\ref{section:eth}}
\label{a0}

\thmsubexp*
\begin{proof}
Due to the assumption that $\Gamma$
  pp-interprets $\sat{3}$, $\Gamma$ can pp-interpret any Boolean
  $\Delta$, as was pointed out in Section~\ref{sec:algcspconj}. In
  particular, $\Gamma$ can pp-interpret the constraint language
  $\{\Rddd\}$ from Jonsson et al.~\cite{jonsson2017}, where $\Rddd =
  \pro_{1,\ldots,6}(\Rdddp)$.
  It is known that SAT$(\Rddd)$-2 is NP-complete and that if it is solvable in subexponential time, then
  the ETH is false~\cite{jonsson2017}.
  Hence, we will
  prove the theorem by giving an LV-reduction from $\SAT(\Rddd)$-2 to
  $\CSP(\Gamma)$, respectively to $\CSP(\Gamma)$-$B$ for some $B > 0$.

  Let $F \subseteq D^{d}$
  and $f : F \mapsto \B$ denote the parameters in the
  pp-interpretation of $\{\Rddd\}$. Note in particular that $d \in \mathbb{N}$ is a
  fixed constant. Let
    \[
    \begin{aligned}
f^{-1}(\Rddd)&(x_{1,1}, \ldots, x_{1,d}, \ldots, x_{6,1}, \ldots,
    x_{6, d}) \equiv \\ & \exists y_1, \ldots, y_{k_1} . \varphi_1(x_{1,1}, \ldots, x_{1,d}, \ldots, x_{6,1}, \ldots,
    x_{6, d}, y_1, \ldots, y_{k_1})
    \end{aligned}
\]
    and
    \[F(x_1, \ldots, x_d) \equiv \exists y_1, \ldots, y_{k_2} . \varphi_2(x_1, \ldots, x_d, z_1,
    \ldots, z_{k_2})\]
    denote efpp-definitions of $f^{-1}(\Rddd)$ and $F$ over
    $\Gamma$ if $\eq_D$ is efpp-definable over $\Gamma$, and otherwise 
    pp-definitions of $f^{-1}(\Rddd)$ and $F$ over $\Gamma$. 
    Let $L$ denote the maximum degree of any variable
    occurring in these pp-definitions, and note that $L$ is a fixed
    constant depending only on $\Gamma$.

    Let $I = (V,C)$ be an instance of $\SAT(\{\Rddd\})$-2. Since each
    variable may occur in at most 2 constraints it follows that $|C|
    \leq 2|V|$. For each
    variable $x_i$ introduce $d$ fresh variables $x_{i,1}, \ldots,
    x_{i,d}$, $k_2$ fresh variables $z_{i,1}, \ldots, z_{i,k_2}$, and
    introduce the constraint \[\varphi_2(x_{i,1}, \ldots, x_{i,d},
    z_{i,1}, \ldots, z_{i,k_2}).\] For each constraint $C_i = \Rddd(x_i, y_i, z_i,
    x'_i,y'_i,z'_i)$ introduce $k_1$ fresh variables $w_{i,1}, \ldots,
    w_{i,k_1}$ and replace $C_i$ by 
    \[\varphi_1(x_{i,1}, \ldots, x_{i,d},
    y_{i,1}, \ldots, y_{i,d}, z_{i,1}, \ldots, z_{i,d}, x'_{i,1}, \ldots,
    x'_{i,d}, y'_{i,1}, \ldots, y'_{i,d}, z'_{i,1}, z'_{i,d}, w_{i,1}, \ldots,
    w_{i,k_1}).\]

    If $\Gamma$ cannot efpp-define $\eq_D$ then we in addition
    identify any two variables occurring in equality constraints. Let $I' = (V',C')$ denote the resulting instance of
    $\CSP(\Gamma)$. Clearly, $I'$ can be constructed in polynomial time.
    We begin by proving that $I'$ has a solution if and only if $I$ has
    a solution. Let $s':V' \rightarrow D$ be a solution to $I'$.
    Recall that every variable $x_i$ in $V$ corresponds to a 'block'
    of variables $x_{i,1},\dots,x_{i,d}$ in $V'$.
    Now, consider a subset $X$ of constraints corresponding to 
    \[\varphi_1(x_{i,1}, \ldots, x_{i,d},
    y_{i,1}, \ldots, y_{i,d}, z_{i,1}, \ldots, z_{i,d}, x'_{i,1}, \ldots,
    x'_{i,d}, y'_{i,1}, \ldots, y'_{i,d}, z'_{i,1}, z'_{i,d}, w_{i,1}, \ldots,
    w_{i,k_1}).\]
    Consider one block of variables $x_{i,1},\dots,x_{i,d}$. We know that
    $(s'(x_{i,1}),\dots,s'(x_{i,d})) \in F$ due to the constraint 
    $F(x_{i,1},\dots,x_{i,d})$ and that $s'$ satisfies $X$. Since $X$ 
    and the block of variables are arbitrarily chosen, we conclude that 
    the function 
    $s:V \rightarrow \B$ defined by
    \[s(x)=f(s'(x_1),\dots,s'(x_d))\]
    is a solution to $I$.

    Assume instead that $s:V \rightarrow \B$ is a solution to $I$. 
    Arbitrarily choose $t_0,t_1 \in F$
    such that $f(t_0)=0$ and $f(t_1)=1$. For each variable $x_i \in V$,
    let $x_{i,1},\dots,x_{i,d}$ denote the corresponding block of variables
    in $V'$, and let $\hat{V}$ denote the set of all these variables.
    Define the function $\hat{s}:\hat{V} \rightarrow F$ such that
    $\hat{s}(x_{i,j})=t_0[j]$ if $s(x_i)=0$ and $\hat{s}(x_{i,j})=t_1[j]$
    otherwise. The function $\hat{s}$ satisifes every constraint
    $F(x_{i,1},\dots,x_{i,d})$ by definition.
    Consider a subset $X$ of constraints corresponding to 
    \[\varphi_1(x_{i,1}, \ldots, x_{i,d},
    y_{i,1}, \ldots, y_{i,d}, z_{i,1}, \ldots, z_{i,d}, x'_{i,1}, \ldots,
    x'_{i,d}, y'_{i,1}, \ldots, y'_{i,d}, z'_{i,1}, z'_{i,d}, w_{i,1}, \ldots,
    w_{i,k_1}).\]
    Recall that $\varphi_1$ is a pp-definition of 
    $f^{-1}(\Rddd)$. Thus, the variables $w_{i,1},\dots,w_{i,k_1}$ can be assigned
    values that in combination with the values provided by $\hat{s}$ satisfies
    $\varphi_1$ and, consequently, $X$. This implies that there is a solution
    to $I'$.

\smallskip

We continue by analysing this reduction.
First, observe that if $\Gamma$ can efpp-define $\eq_D$ then the maximum degree of any
    variable is $3L$. This implies that $I'$ is in fact an instance of $\CSP(\Gamma)$-$3L$.
    Second, note that $|C| \leq 2|V|$, and that we for
    every constraint in $C$ introduce $k_1$ fresh variables. This implies
    that $|V'| \leq |V|d + 2|V|k_1 + k_2$, and, since $k_1$, $k_2$ and $d$ are fixed
    constants, there exists a constant $K$ such that $|V'| = K|V| +
    O(1)$. Since this reduction is an LV-reduction from $\SAT(\Rddd)$-2 to
    $\CSP(\Gamma)$-$3L$ (or to $\CSP(\Gamma)$ if $\Gamma$ cannot
    efpp-define $\eq_D$), it follows that $\SAT(\Rddd)$-2 is solvable in
    subexponential time if $\CSP(\Gamma)$-$3L$ (or $\CSP(\Gamma)$) is solvable in subexponential
    time.
\end{proof}
\section{Additional Proofs for Section~\ref{section:easiest_csp}}
\label{a1}

We will need the following lemma before we can present the proof for Lemma~\ref{lem:easiest_ultra}.

\begin{lemma1} \label{lem:2tuples}
  Let $\Gamma$ be an ultraconservative language over a finite domain
  $D$ and let $R \in
  \cclone{\Gamma}$ be an $n$-ary relation such that $|R| = 2$. Then
  there exists $R' \in \pcclone{\Gamma}$ such that (1) $|R'| = 2$ and
  (2) $\pro_{1, \ldots, n}(R') = R$.
\end{lemma1}

\begin{proof}
  Let $R(x_1, \ldots, x_n) \equiv \exists y_1, y_2, \ldots, y_m . \varphi(x_1,
  \ldots, x_n, y_1, y_2, \ldots, y_m)$ denote a pp-definition of $R$ over
  $\Gamma$, and let $R = \{t_1, t_2\}$. We will show that it is possible to remove the
  existentially quantified arguments $y_1, y_2, \ldots, y_m$ in this
  pp-definition by gradually adding new arguments to $R$. First consider the relation $R_1(x_1, \ldots, x_n, y_1) \equiv
  \exists y_2 \ldots, y_m . \varphi(x_1,
  \ldots, x_n, y_1, y_2, \ldots, y_m)$. If $|R_1| = 2$ then we move on with
  the remaining arguments, so instead assume that $|R_1| > 2$. Now
  note that each tuple $t \in R_1$ in a natural way can be associated
  with either $t_1 \in R_1$ or $t_2 \in R_2$, depending on whether $t
  = t_1^{\frown}t'$ or $t = t_2^{\frown}t'$. Hence, let $S_1 = \{t[n+1] \mid t \in R_1, t_1^{\frown}t'
  = t\}$, and $S_2 = \{t[n+1] \mid t \in R_1, t_2^{\frown}t'
  = t\}$. In other words $S_1$ is the set of values taken by $y_1$ in the
  tuples corresponding to $t_1$, and $S_2$ the values taken by $y_1$
  in the tuples corresponding to $t_2$. 
  We consider two cases.

 \smallskip

 \noindent
 {\em Case 1:} $S_1 \cap S_2 = \emptyset$. Arbitrarily choose $d_1 \in S_1$ and
 $d_2 \in S_2$. Construct
  the relation $R'_1(x_1, \ldots, x_n, y_1) \equiv R_1(x_1, \ldots,
  x_n, y_1) \land \{(d_1), (d_2)\}(y_1)$, and note that $\{(d_1),
  (d_2)\} \in \Gamma$ since $\Gamma$ is ultraconservative. We see that
  $R'_1=\{s_1^{\frown}(d_1),s_2^{\frown}(d_2)\}$.

\smallskip

 \noindent
 {\em Case 2: }$S_1 \cap S_2 \neq \emptyset$. Arbitrarily choose $d \in S_1 \cap S_2$
and construct
  the relation $R'_1(x_1, \ldots, x_n, y_1) \equiv R_1(x_1, \ldots,
  x_n, y_1) \land \const{d}(y_1)$. We see that 
  $R'_1=\{s_1^{\frown}(d),s_2^{\frown}(d)\}$. Note that we cannot choose
  elements as in Case 1 since if (for instance) one element is inside $S_1 \cap S_2$
  and one element is outside $S_1 \cap S_2$, then the resulting relation
  will contain three tuples.

\smallskip

  If we repeat this procedure for the
  remaining arguments $y_2, \ldots, y_m$ we will obtain a relation
  $R'$ which is qfpp-definable over $\Gamma$ such that $|R'| = 2$ and $\pro_{1, \ldots, n}(R') = R$.
\end{proof}

\lemeasiestultra*

\begin{proof}
  By Lemma~\ref{lem:ppi} there exists a relation $R \in
  \cclone{\Gamma}$ which is an $\RB$-extension. Let 
  \[R(x_1, \ldots, x_n) \equiv \exists y_1, y_2, \ldots,
  y_m . \varphi(x_1, \ldots, x_n, y_1, \ldots, y_m)\] denote its
  pp-definition over $\Gamma$. Using this pp-definition we will show
  that $\Gamma$ can qfpp-define an $\RB$-extension by gradually
  removing each existentially quantified variable. First consider the relation $R_1(x_1,
  \ldots, x_n, y_1) \equiv \exists y_2, \ldots, y_m . \varphi(x_1,
  \ldots, x_n, y_1, y_2, \ldots, y_m)$. Assume that $|R_1| > 3$, i.e., that $R_1$
  is not an $\RB$-extension. 
  Let $R = \{t_1, t_2, t_3\}$ and for each $1 \leq i \leq 3$ let $S_i = \{t[n+1] \mid t \in R_1, t_i^{\frown} t'
  = t\}$, $1 \leq i \leq 3$. In other words $S_i$ contains the possible values taken by the
  argument $y_1$ in the tuples of $R_1$ corresponding to $t_i \in R$. There
  are now a few cases to consider depending on the sets $S_1, S_2,
  S_3$:

  \begin{enumerate}
  \item
    $|S_1 \cup S_2 \cup S_3| = 1$,
  \item
    $|S_1 \cup S_2 \cup S_3| = 2$, and
  \item
    $|S_1 \cup S_2 \cup S_3| \geq 3$,
  \end{enumerate}
  
  The first case implies that the $(n+1)$th argument of $R_1$ is
  constant and that $R_1$ is already an $\RB$-extension. In the third case,
  first choose $d_1 \in S_1$. If $d_1 \in S_2$ then let $d_2 = s_1$,
  otherwise choose an arbitrary value in $S_2$ distinct from
  $d_1$. Last, if $d_1 \in S_3$ or $d_2 \in S_3$ then let $d_3 = d_1$
  or $d_3 = d_2$; otherwise choose an arbitrary value not occurring in
  $S_1 \cup S_2$. Note that this is possible since we assumed that
  $|S_1 \cup S_2 \cup S_3| \geq 3$, which implies that $S_1 \cup S_2
  \cup S_3$ contains at least three distinct values. Let $E$ be the
  unary relation $\{(d_1), (d_2), (d_3)\}$. It is then easy
  to see (by basically reasoning in the same way as in the proof of
  Lemma~\ref{lem:2tuples}) that 
  $\exists y_2, \ldots, y_m . E(y_1) \land \varphi(x_1,
  \ldots, x_n, y_1, \ldots, y_m)$ defines an $\RB$-extension.

  Now assume that $|S_1 \cup S_2 \cup S_3| = 2$ and let $\{d_1, d_2\}
  = S_1 \cup S_2 \cup S_3$. Up to symmetry, we then have the following
  possible cases:

  \begin{enumerate}
  \item
    $S_1 = S_2 = S_3 = \{d_1, d_2\}$,
  \item
    $S_1 = S_2 = \{d_1, d_2\}$, $S_3 = \{d_1\}$, or
  \item
    $S_1 = \{d_1\}$, $S_2 = \{d_2\}$, $S_3 = \{d_1, d_2\}$.
  \end{enumerate}

  The first two cases are easy to handle in a similar way to the case
  when $|S_1 \cup S_2 \cup S_3| \geq 3$; in both cases, choose the element 
  $d_1$.
  This
  leaves only the case when $S_1 = \{d_1\}$, $S_2 = \{d_2\}$ and that
  $S_3 = \{d_1, d_2\}$. 
  Since $R$ is an $\RB$-extension there
  exists $a,b \in D$, $a \neq b$, and indices $i_1, i_2, i_3$ such that 
  $(t_1[i_1],t_2[i_1], t_3[i_1]) = (b,b,a)$, 
  $(t_1[i_2],t_2[i_2], t_3[i_2]) = (b,a,b)$, and 
  $(t_1[i_3],t_2[i_3], t_3[i_3]) = (a,b,b)$.
  Define the binary relation $F$ such that
  \[F(x, y_1) \equiv \exists x_1, \ldots x_{i_3 - 1}, x_{i_3 + 1},
  \ldots, x_n
  . R_1(x_1, \ldots, x_{i_3 - 1}, x, x_{i_3 + 1}, \ldots, x_n, y_1) \land \const{b}(x_{i_1}).\]
We claim that $F=\{(a,d_1),(b,d_2)\}$. To see this, observe that the constraint $\const{b}(x_{i_1})$
  rules out the tuple $t_3$. This implies that if variable $x_{i_3}$ has value
  $a$, then the variable $y_1$ must have value $d_1$ and if the variable $x_{i_3}$
  has value $b$, then the variable $y_1$ must have value $d_2$.

  From this observation and 
  Lemma~\ref{lem:2tuples}, it follows that $\Gamma$ can qfpp-define a
  relation $F'$ such that $|F'| = 2$ and such that $\pro_{1,2}(F') =
  F$. Let $k + 2$ denote the arity of $F'$ and define a relation
  \[R'_1(x_1, \ldots, x_{i_3}, \ldots, x_n, y_1, z_1, \ldots,
  z_k) \equiv R_1(x_1, \ldots, x_{i_3}, \ldots, x_n, y_1) \land F'(x_{i_3}, y_1, z_1,\ldots, z_k).\] 
  We claim that $R'_1$ is an $\RB$-extension. There are three possible
  ways of simultaneously choosing variables $x_{i_1},x_{i_2},x_{i_3}$. Let us
  consider the assignment
  $(x_{i_1},x_{i_2},x_{i_3})=(b,b,a)$. This particular choice gives
  all variables $x_1,\dots,x_n$ fixed values (via the constraint
  $R_1(x_1,\ldots, x_{i_3}, \ldots, x_n, y_1)$).
 Furthermore,  $y_1$ is assigned the
  value $d_2$ (via the constraint $F'(x_{i_3},y_1,z_1,\dots,z_k)$) and the 
  variables $z_1,\dots,z_k$ are given fixed values (since there is only one
tuple in $F'$ that allows $y_1$ to have the value $d_2$).
  Thus, there is only one tuple in $R'_1$ that allows $(x_{i_1},x_{i_2},x_{i_3})=(b,b,a)$. The two other cases can be verified similarly and we conclude
  that $|R'_1|=3$.
 
  Finally, we see that there are $m-1$ existentially quantified variables in
  the definition of $R'_1$ since $F'$ can be qfpp-defined.
  By repeating the procedure outline above for the remaining arguments
  we will obtain an $\RB$-extension which is qfpp-definable over
  $\Gamma$. This concludes the proof.
\end{proof}

Before the proof of Lemma~\ref{lem:2saturated} we will need the following result from
Lagerkvist et al.~\cite[Lemma 2]{lagerkvist2015}, restated in slightly
simpler terminology.

\begin{lemma1}\label{lem:contraction}
  Let $R$ be a relation with $m$ tuples. If $f \notin \ppol(R)$, where
  $f$ has arity $n > m$, there exists $g$ of arity $n' \leq m$ such
  that $g \notin \ppol(R)$ and $g$ can be obtained from $f$ by
  identifying arguments.
\end{lemma1}

For a $k$-ary relation $R$ and tuples
$t_1, \ldots, t_n \in R$ we write $\setcolumns(t_1, \ldots, t_n)$ for
the set $\{(t_1[1], \ldots, t_{n}[1]), \ldots,
(t_{n}[k], \ldots, t_{n}[k])\}$. 

\lemtwosaturated*

\begin{proof}
  Let $R = \{t_1, t_2, t_3\}$ and let $n$ denote the arity of $R$. For
  each $1 \leq i \leq n$ and
  each function $\tau : \{1, 2, 3\} \rightarrow \{1, 2, 3\}$ add a
  fresh argument taking the values $t_{\tau(1)}[i]$, $t_{\tau(2)}[i]$, $t_{\tau(3)}[i]$.
  Let $R'$ be the resulting
  relation and let $R' = \{t'_1, t'_2, t'_3\}$ such that $\pro_{1,
    \ldots, n}(t'_{i}) = t_{i}$. 
  By construction, $R'$ is a saturated $\RB$-extension,
  but it remains to prove that $R' \in \pcclone{R}$. Hence, assume
  with the aim of reaching a contradiction, that $R' \notin
  \pcclone{R}$. Due to the Galois connection in Theorem~\ref{theorem:galois} this
  implies that $\ppol(R) \not \subseteq \ppol(R')$. Hence, there exists a
  partial function $f$ preserving $R$ but which does not preserve $R'$, and
  due to Lemma~\ref{lem:contraction} we may without loss of
  generality assume that $f$ has arity at most 3. We omit
  the cases when $\ar(f) \leq 2$ since they are similar, and therefore
  assume that
  $f(t'_{\rho(1)}, t'_{\rho(2)}, t'_{\rho(3)}) = t' \notin R'$ for a
  permutation $\rho$ on $\{1,2,3\}$. Note that since $\pro_{1, \ldots,
  n}(R') = R$ it must hold that $\setcolumns(t_{\rho(1)},
  t_{\rho(2)}, t_{\rho(3)}) \subseteq \setcolumns(t'_{\rho(1)}, t'_{\rho(2)},
  t'_{\rho(3)})$. Hence, $f(t_{\rho(1)},
  t_{\rho(2)}, t_{\rho(3)})$ must be defined, and furthermore $f(t_{\rho(1)}, t_{\rho(2)},
  t_{\rho(3)}) \in R$ since we assumed that $f$ preserves $R$. Assume without loss of generality that $f(t_{\rho(1)}, t_{\rho(2)},
  t_{\rho(3)}) = t_{\rho(1)}$, i.e., $f$ restricted to the tuples $t_{\rho(1)}, t_{\rho(2)},
  t_{\rho(3)}$ is a projection on the first
  argument. Since $f$ when applied to $t'_{\rho(1)}, t'_{\rho(2)}, t'_{\rho(3)}$
  by assumption is {\em not} a projection, there
  exists at least one index $j \in \{n+1, \ldots, \ar(R')\}$ such that $f(t'_{\rho(1)}[j],
  t'_{\rho(2)}[j], t'_{\rho(3)}[j]) \neq t'_{\rho(1)}[j]$. 
Due to the construction of $R'$, there exists $i \in \{1, \ldots, n\}$
and a function $\tau' : \{1, 2, 3\} \rightarrow \{1, 2, 3\}$ such that \[(t_{\tau'(1)}[i], t_{\tau'(2)}[i],
  t_{\tau'(3)}[i]) = (t'_{\tau(1)}[j], t'_{\tau(2)}[j],
  t'_{\tau(3)}[j]).\] In other words it is possible to order the
  tuples from $R$ in such a way that
  the values enumerated by these tuples in position $i$ is exactly equal to $(t'_{\tau(1)}[j], t'_{\tau(2)}[j],
  t'_{\tau(3)}[j])$, where $f$ is not a projection. It follows that $\setcolumns(t_{\tau'(1)}, t_{\tau'(2)},
  t_{\tau'(3)}) \subseteq \setcolumns(t'_{\tau(1)}, t'_{\tau(2)},
  t'_{\tau(3)}) \subseteq \domain(f)$ (since $R'$ is saturated) and therefore also that $f(t_{\tau'(1)}, t_{\tau'(2)},
  t_{\tau'(3)}) \notin R$ (since $f$ is not a projection on these tuples). This contradicts the assumption that $f \in
  \ppol(R)$, and it must therefore be the case that $R' \in \pcclone{R}$.
\end{proof}

\lemthreechoice*

\begin{proof}
  Let $n$ denote the arity of $R$ and let $\{t_1, t_2, t_3\} = R$.  Let
  $I = (V,C)$ be an instance of $\CSP(R)$. We will create an instance
  $I' = (V', C')$ of $\CSP(R)$ such that if $x \in V'$ is a 3-choice
  variable in a constraint then $x$ does not occur as a
  3-choice variable in any other constraint. Hence, let $x \in V$ be a
  3-choice variable occurring in a constraint $c = R(x_1, \ldots, x_n)$
  in position $i_1$. Assume that $x$ also appears as a 3-choice variable
  in a constraint $c' = R(x'_1, \ldots, x'_n)$, distinct from $c$, in
  position $i_2$. Let $S = (t_1[i_1], t_2[i_1], t_3[i_1])$ and $S' =
  (t_1[i_2], t_2[i_2], t_3[i_2])$.
  
  Assume first that $\pro_{i_1}(R) =
  \pro_{i_2}(R)$. Define the function $\tau$ such that for
  each $1 \leq i \leq 3$, $\tau(S[i]) = j$ if and only if $t_j[i_2] =
  S[i]$ where $1 \leq j \leq 3$. Using the function $\tau$ we then define the permutation
  $\rho : \{1, \ldots, n\} \rightarrow \{1, \ldots, n\}$ such that
  $\rho(i) = j$ if and only if $(t_1[i], t_2[i], t_3[i]) =
  (t_{\tau(1)}[j], t_{\tau(2)}[j], t_{\tau(3)}[j])$.
  This is indeed a
  well-defined permutation over $\{1, \ldots, n\}$ since $R$ is
  saturated. Last, identify each variable
  $x'_{\tau(i)}$ occurring in $c'$ with the variable $x_i$ in $c$, and
  remove the constraint $c'$. 

  Second, assume that $|\pro_{i_1}(R) \cap \pro_{i_2}(R)| = 2$, and
  let $\pro_{i_1}(R) \cap \pro_{i_2}(R) = \{d, d'\}$. Assume without
  loss of generality that $t_1[i_1] = d$, $t_2[i_1] = d'$, and that
  $t_3[i_1] \notin \{d, d'\}$. Choose $i \in
  \{1, \ldots, n\}$, distinct from both $i_1$ and $i_2$, such that
  $t_1[i] = t_1[i_1]$, $t_2[i] = t_2[i_1]$, and $t_3[i] \neq
  t_3[i_1]$. Such an $i$ must exist since $R$ is saturated. Then
  identify $x$ with $x_i$. Define the function $\tau$ such that for
  $1 \leq i \leq 2$, $\tau(S[i]) = j$ if and only if $t_j[i_2] =
  S[i]$. Using the function $\tau$ we then define the permutation
  $\rho : \{1, \ldots, n\} \rightarrow \{1, \ldots, n\}$ such that
  $\rho(i) = j$ if and only if $(t_1[i], t_2[i]) =
  (t_{\tau(1)}[j], t_{\tau(2)}[j])$. Clearly, $\tau$ is a
  well-defined permutation over $\{1, \ldots, n\}$ since $R$ is
  saturated. Last, identify each variable
  $x'_{\tau(i)}$ occurring in $c'$ with the variable $x_i$ in $c$, and
  remove the constraint $c'$. The case when $|\pro_{i_1}(R) \cap
  \pro_{i_2}(R)| = 1$, i.e., when $x$ is assigned the same value in
  any satisfying assignment, is very similar.

  Each time this procedure is performed, at least one constraint is
  removed. Thus, we let $I'$ denote the fixpoint that we will reach in
  at mots $|C|$ iterations.
  It is not difficult to verify that $I$ is satisfiable if and only if
  $I'$ is satisfiable. Furthermore, $|V'| \leq |V|$ and
  the reduction can be computed
  in polynomial time. We have thus showed that the reduction is a
  CV-reduction and therefore proved the lemma.
\end{proof}

\lemtwochoice*

\begin{proof}
Let $n=\ar(R)$, $n'=\ar(R')$, and
$R'=\{t'_1,t'_2,t'_3\}$. By the statement of the lemma we may assume
that $\pro_{1, \ldots, n}(R') = R$, and that $|\pro_{i}(R')| = 2$ for
every $n' < i \leq n$. We will furthermore assume that $\pro_{i}(R')$
for every $n' < i \leq n$
is distinct from $\pro_{j}(R')$ for every $1 \leq j \leq n$.
To simplify the proof we also assume that 
$\pro_{1,\ldots ,8}(R')=\RB$. 
Let $I=(V,C)$ be an instance of $\CSP(R')$.
Let $x$ be a variable that appears in two distinct constraints
$c_1,c_2\in C$. Assume that $x$
occurs at position $n+1 \leq i \leq n'$ in $c_1$ and
at position $1 \leq j \leq n'$ in $c_2$. 
We consider a number of cases
based on the cardinality of $S=\pro_i(R')\cap \pro_j(R')$.
\begin{itemize}
\item $|S|=3$.
This is not possible since $|\pro_i(R')| = 2$.

\item $|S|=2$.
Assume that $S=\{a,b\}$ and $\pro_j(R')=\{a,b,d\}$ (where $b,d$ are not
necessarily distinct). Define $f:\{a,b\}\rightarrow \{0,1\}$ such that
$f(a)=0$ and $f(b)=1$ and $g: \{a,b,d\}\rightarrow \{0,1\}$ such that
$g(a)=0$ and $g(x)=1$ if $x\neq a$.  It follows that there exist
indices $l,m \in \{1,\dots,6\}$ such that $f(t'_r[i])=t'_r[l]$ and
$g(t'_r[j])=t'_r[m]$ when $r \in \{1,2,3\}$. If $b\neq d$, then we
need ensure that $x$ is never assigned $d$ in any satisfying
assignment to the resulting instance. 
For simplicity, assume that $t'_1[j]=d$. Then there exists $p\in
\{1,\ldots, 6\}$ such that $ t'_1[p]=1, t'_2[p]=0, t'_3[p]=0$. Let $w$
be the variable at position $p$ in $c_1$, and add the unary relation
$\const{0}(w)$. 
Now, let $y$ be the variable at position $l$ in $c_1$ and let $z$ be the
variable at position $m$ in $c_2$. The variable $x$ implies that $y,z$
will always be assigned the same value by a solution to $I$. Hence, we
identify $z$ with $y$, introduce a fresh variable $\hat{x}$, and
replace $x$ at the $i$th position of $c_1$ with $\hat{x}$.

\item $|S|=1$.
Assume $S=\{a\}$, $\pro_i(R')=\{a,b\}$ (where $a,b$ are distinct elements),
and $\pro_j(R')=\{a,d,d'\}$ (where $a,d,d'$ are not necessarily distinct).
Define $f:\{a,b\}\rightarrow \{0,1\}$ such that $f(a)=0$ and $f(b)=1$,
and $g:\{a,d,d'\} \rightarrow \{0,1\}$
such that $g(a)=0$ and $g(x)=1$ if $x \neq a$.
It is not hard to see that there exists $l,m \in \{1,\dots,8\}$ such that
$f(t'_r[i])=t'_r[l]$ and $g(t'_r[j])=t'_r[m]$ when $r \in \{1,2,3\}$.
Let $y$ be the variable at position $l$ in $c_1$ and $z$ be the variable
at position $m$ in $c_2$.
Add the unary relations $\const{0}(y)$ and $\const{0}(z)$,
introduce a new variable
$\hat{x}$, and replace $x$ at the $i$th position of $c_1$ with $\hat{x}$.

\item $|S|=0$.
This implies $I_1$ is unsatisfiable, and we simply output 
an arbitrary unsatisfiable instance. 
\end{itemize}

By repeating the procedure above until a fixpoint is reached, we will
obtain an instance $I_1=(V_1,C_1)$ such that if $x\in V_1$ and if $x$
appears in a constraint $c\in C_1$ at position $n+1, \ldots, n'$, then
it does not appear in any other constraint. 
However, it is still possible that $x\in V_1$ appear more than once in
a single constraint $c\in C_1$ where (at least) one of the occurrences
of $x$ is at position $n+1,\dots,n'$.  Therefore, assume that $x$
appears in positions $i$ and $j$ in $c\in C_1$ where $i\in
\{n+1,\ldots n'\}$ and $j\in \{1,\ldots n'\}$.  Let $L\subseteq
\{1,2,3\}$ denote the set $\{l \; | \; t'_l[i]=t'_l[j]\}$.

\begin{itemize}
  
\item $|L|=3$.
This is not possible since there are no redundant arguments in the relation 
$R'$.

\item $|L|=2$.
Assume (without loss of generality) that $t'_1[i]=t'_1[j]$, 
$t'_2[i]=t'_2[j]$, and
$t'_3[i] \neq t'_3[j]$. Pick $k \in \{1, \ldots, 8\}$ such that
$t'_1[k]=t'_2[k] \neq t'_3[k]$. Let $y$ be the 
variable that appear in the $k$th position in $c$. Add a 
unary constraint $\const{t'_1[k]}(y)$, introduce a fresh variable $\hat{x}$, 
and replace the $x$ at position $i$ in $c$ with $\hat{x}$.

\item $|L|=1$.
Without loss of generality we can assume that
$t'_1[i]=t'_1[j]$. For each variable $y$ occurring in the $l$th
position in $c$ add the unary constraint $\const{t'_1[l]}(y)$, and
then remove the constraint $c$.

\item $|L|=0$.
This implies that $I_1$ is unsatisfiable, and we simply output
an arbitrary unsatisfiable instance. 
\end{itemize}

Repeat the procedure above until a fixpoint is reached and let
$I_2=(V_2,C_2)$ be the resulting instance. 
Observe that a variable $x$ that occurs in a constraint at position
$n+1,\dots,n'$ only occur in a single constraint and in a unique
position.  Finally, let $I_3=(V_3,C_3)$ be the instance of $\CSP(R)$
obtained by replacing each constraint $R'(x_1, \ldots, x_n, x_{n+1},
\ldots, x_{n'}) \in C_2$ by $R(x_1, \ldots, x_n)$.  Note that every
fresh variable $\hat {x}$ that were introduced in the previous steps
are removed in the conversion of $I_2$ into $I_3$. This shows that the
reduction is indeed a CV-reduction.
\end{proof}

\lemthreechoicethree*
\begin{proof}
  Let $n = \ar(R)$. Choose three distinct values $d_1, d_2, d_3 \in D$ such that there
  does not exist any $i$ such that $\pro_i(R) = \{d_1, d_2, d_3\}$. If
  no such $i$ exists then $\pcclone{R} = \pcclone{\RD{D}}$, and we are
  done. First, construct the relation $S$ such that $\pro_{1, \ldots, n}(S) =
  R$, $\pro_{n+1}(S) = \{d_1, d_2, d_3\}$, and then add the minimum
  number of arguments to make $S$ saturated. Second, let $S'$ be the
  relation obtained from $S$ by projecting away every argument $i$ of
  the form $\pro_{i}(S) = \{d_1, d_2, d_3\}$.
  In other words, $S'$ is
  equivalent to $R$, except that it potentially contains more 1-choice and 2-choice
  arguments. Note that $S'$ is saturated. Via Lemma~\ref{lem:2choice} it then follows that
  $\CSP(S') \reduces \CSP(R)$, and an application of Lemma~\ref{lem:3choice2}
  gives the desired result that $\CSP(S) \reduces \CSP(S') \reduces
  \CSP(R)$. This procedure can be repeated
  arbitrarily many times, which implies that $\CSP(\RD{D}) \reduces
  \CSP(R)$.
\end{proof}

\thmmainresultext*

\begin{proof}
  ${\sf T}(\{\RD{D}\}) \leq {\sf
    T}(\{\RD{D}\} \cup 2^D)$ holds trivially. To prove ${\sf
    T}(\{\RD{D}\} \cup 2^D) \leq {\sf T}(\{\RD{D}\})$ we show that
  $\CSP(\{\RD{D}\} \cup 2^D) \reduces \CSP(\RD{D})$. Since we have
  already seen many reductions akin to this we only provide a sketch. Let $(V,C)$ be an
  instance of $\CSP(\{\RD{D}\} \cup 2^D)$. Assume $x \in V$ appears in
  a unary constraint $E(x) \in C$. If $x$ also appears in another
  unary constraint $E'(x)$ then these two constraints can be replaced
  by $E \cap E'(x)$; hence, we may assume that each variable occurs in
  at most one unary constraint. If $x$ does not occur in any other
  constraint, then we first check if $E = \emptyset$. If this is the
  case, the instance is unsatisfiable and we abort the procedure, and
  otherwise we simply remove the constraint $E(x)$. Now assume that
  $x$ also appears in the $i$th position in a constraint $\RD{D}(x_1,
  \ldots, x_{i-1}, x, x_{i+1}, \ldots, x_{\ar(\RD{D})})$. If $E \cap
  \pro_i(\RD{D}) = \emptyset$ then the instance is unsatisfiable, and
  if $E = \pro_{i}(\RD{D})$ then we may safely remove the constraint
  $E$. Therefore assume that either $|\pro_i(\RD{D}) \cap E| = 1$ or
  that $|\pro_i(\RD{D}) \cap E| = 2$. The first of these cases is easy
  to handle since it implies that $x$ is forced a constant value in
  any satisfying assignment. The second case implies that $x$ appears
  in a 3-choice position, i.e., $\pro_i(\RD{D}) = \{d_1, d_2, d_3\}$,
  for three distinct values $d_1, d_2$, and $d_3$. Assume that $E =
  \{(d_1), (d_2)\}$, and let $t \in \RD{D}$ be the tuple satisfying
  $t[i] = d_3$. Let $\{s,u\} = \RD{D} \setminus \{t\}$ and choose $j$
  such that $s[j] = s[i]$, $u[i] = u[j]$, and $t[j] \in
  \{s[j],u[j]\}$. Then identify $x$ with the variable $x_j$ throughout
  the instance.  If we repeat this procedure for the remaining constraints
  containing $x$, remove the constraint $U(x)$, and then continue with all remaining unary constraints, we will
  obtain an instance of $\CSP(\RD{D})$ which is satisfiable if and
  only if $(V,C)$ is satisfiable.
\end{proof}

\end{document}